\documentclass[11pt]{article}
\usepackage{amsmath,amsthm,amssymb}
\usepackage{fullpage}
\usepackage[onehalfspacing]{setspace}
\usepackage[parfill]{parskip}
\usepackage[dvipsnames]{xcolor}
\usepackage{newpxtext}
\usepackage{newpxmath}
\usepackage{tikz}
\usepackage{mathtools}
\usepackage{bm}
\usepackage{todonotes}
\usepackage{array}
\usepackage{makecell}
\usepackage{booktabs,tabularx,threeparttable}
\usepackage{hyperref}
\usepackage[nameinlink, capitalise]{cleveref-usedon}

\newtheorem{theorem}{Theorem}[section]
\newtheorem{lemma}[theorem]{Lemma}

\newtheorem{corollary}[theorem]{Corollary}
\newtheorem{proposition}[theorem]{Proposition}

\theoremstyle{definition}
\newtheorem{definition}[theorem]{Definition}

\crefname{theorem}{Theorem}{Theorems}
\crefname{lemma}{Lemma}{Lemmas}
\crefname{conjecture}{Conjecture}{Conjectures}
\crefname{corollary}{Corollary}{Corollaries}
\crefname{proposition}{Proposition}{Propositions}
\crefname{definition}{Definition}{Definitions}
\crefname{example}{Example}{Examples}
\crefname{remark}{Remark}{Remarks}
\crefname{question}{Question}{Questions}
\crefname{condition}{Condition}{Conditions}

\AddToHook{env/theorem/begin}{\crefalias{theorem}{theorem}}
\AddToHook{env/lemma/begin}{\crefalias{theorem}{lemma}}
\AddToHook{env/conjecture/begin}{\crefalias{theorem}{conjecture}}
\AddToHook{env/corollary/begin}{\crefalias{theorem}{corollary}}
\AddToHook{env/proposition/begin}{\crefalias{theorem}{proposition}}
\AddToHook{env/definition/begin}{\crefalias{theorem}{definition}}
\AddToHook{env/example/begin}{\crefalias{theorem}{example}}
\AddToHook{env/remark/begin}{\crefalias{theorem}{remark}}
\AddToHook{env/question/begin}{\crefalias{theorem}{question}}
\AddToHook{env/condition/begin}{\crefalias{theorem}{condition}}

\crefname{Program}{Program}{Programs}
\creflabelformat{Program}{(#2\textup{#1})#3}

\renewcommand*\d{\mathop{}\!\mathrm{d}}
\newcommand{\E}{\mathbb{E}}

\definecolor{linkc}{rgb}{0.6, 0.2, 0.3}
\definecolor{citec}{rgb}{0.3, 0.2, 0.6}
\definecolor{urlc}{rgb}{0.3, 0.45, 0.3}
\hypersetup{
    colorlinks=true,
    linkcolor=linkc,
    citecolor=citec,
    urlcolor=urlc
}

\title{Winning in the Limit:\\ Average-Case Committee Selection with Many Candidates}
\author{\begin{tabular}{cc}
\begin{minipage}[t]{0.4\textwidth}\centering
Yifan Lin\\
\small Shanghai Jiao Tong University\\
\small\href{mailto:sjtu15088284174@sjtu.edu.cn}{sjtu15088284174@sjtu.edu.cn}
\end{minipage}
&
\begin{minipage}[t]{0.4\textwidth}\centering
Shenyu Qin\\
\small Shanghai Jiao Tong University\\
\small\href{mailto:celery2022@sjtu.edu.cn}{celery2022@sjtu.edu.cn}
\end{minipage}
\\[7ex]
\begin{minipage}[t]{0.4\textwidth}\centering
Kangning Wang\\
\small Rutgers University\\
\small\href{mailto:kn.w@rutgers.edu}{kn.w@rutgers.edu}
\end{minipage}
&
\begin{minipage}[t]{0.4\textwidth}\centering
Lirong Xia\\
\small DIMACS, Rutgers University\\
\small\href{mailto:lirong.xia@rutgers.edu}{lirong.xia@rutgers.edu}
\end{minipage}
\end{tabular}}
\date{}

\begin{document}

\maketitle
\thispagestyle{empty}
\setcounter{page}{0}

\begin{abstract}
We study the committee selection problem in the canonical impartial culture model with a large number of voters and an even larger candidate set. Here, each voter independently reports a uniformly random preference order over the candidates. For a fixed committee size $k$, we ask when a committee can collectively beat every candidate outside the committee by a prescribed majority level $\alpha$. We focus on two natural notions of collective dominance, $\alpha$-winning and $\alpha$-dominating sets, and we identify sharp threshold phenomena for both of them using probabilistic methods, duality arguments, and rounding techniques.

We first consider \emph{$\alpha$-winning sets}. A set $S$ of $k$ candidates is $\alpha$-winning if, for every outside candidate $a \notin S$, at least an $\alpha$-fraction of voters rank \emph{some} member of $S$ above $a$. We show a sharp threshold at
\[
\alpha_{\mathrm{win}}^\star = 1 - \frac{1}{k}.
\]
Specifically, an $\alpha$-winning set of size $k$ exists with high probability when $\alpha < \alpha_{\mathrm{win}}^\star$, and is unlikely to exist when $\alpha > \alpha_{\mathrm{win}}^\star$.

We then study the stronger notion of \emph{$\alpha$-dominating sets}. A set $S$ of $k$ candidates is $\alpha$-dominating if, for every outside candidate $a \notin S$, there exists a \emph{single} committee member $b \in S$ such that at least an $\alpha$-fraction of voters prefer $b$ to $a$. Here we establish an analogous sharp threshold at
\[
\alpha_{\mathrm{dom}}^\star = \frac{1}{2} - \frac{1}{2k}.
\]
That is, an $\alpha$-dominating set of size $k$ exists with high probability when $\alpha < \alpha_{\mathrm{dom}}^\star$, and is unlikely to exist when $\alpha > \alpha_{\mathrm{dom}}^\star$. As a corollary, our analysis yields an impossibility result for $\alpha$-dominating sets: for every $k$ and every $\alpha > \alpha_{\mathrm{dom}}^\star = 1 / 2 - 1 / (2k)$, there exist preference profiles that admit no $\alpha$-dominating set of size $k$. This corollary improves the best previously known bounds for all $k \geq 2$.
\end{abstract}
 
\newpage

\section{Introduction}
Consider the following canonical setting in social choice theory: there are $n$ agents (voters) and $m$ alternatives (candidates); each agent's vote is a full ranking of the alternatives, and the goal is to choose a single winner from the alternatives. In this setting, the Condorcet criterion~\cite{Condorcet1785:Essai} is one of the most classical and well-accepted principles: an alternative $a^*$ is a \emph{Condorcet winner} if it beats every other alternative $b$ in their head-to-head competition, and the Condorcet criterion requires the winner to be the Condorcet winner whenever it exists. 

Unfortunately, the \emph{Condorcet paradox}~\cite{Condorcet1785:Essai} points out that a Condorcet winner does not always exist. As a simple example, let there be three alternatives $\{a,b,c\}$ and three votes (1) $a\succ b\succ c$, (2) $b\succ c\succ a$, and (3) $c\succ a\succ b$, where the symbol ``$\succ$'' means ``is preferred to.'' It follows that in head-to-head comparisons, $a$ beats $b$, $b$ beats $c$, and $c$ beats $a$, forming intransitive social preferences known as the \emph{Condorcet cycle}. Subsequent studies to understand and address such intransitivity have shaped the field of social choice theory. For example, the paradox was extended by Arrow in his celebrated impossibility theorem~\cite{Arrow63:Social} and can be circumvented by \emph{domain restriction}, requiring agents' preferences to be from a proper subset of all full rankings~\cite{Black48:Rationale,Sen1966:A-Possibility}. Various relaxations of the Condorcet criterion have been considered in the following three directions. 
\begin{itemize}
\item \textbf{Direction 1: Relax the strength of social preferences.} Given a preference profile and a parameter $\alpha \in [0,1]$, we say that an alternative $a$ \emph{$\alpha$-beats} another alternative $b$ if at least $\alpha n$ agents prefer $a$ to $b$. An \emph{$\alpha$-Condorcet winner}, therefore, is an alternative that $\alpha$-beats all other alternatives~\cite{Sertel2004:Strong}. Note that an alternative is a $0.5$-Condorcet winner if and only if it is a standard Condorcet winner. The maximin rule (also known as Simpson's rule or Kramer's rule) chooses an $\alpha$-Condorcet winner with the largest $\alpha$.
\item \textbf{Direction 2: Relax the number of winners.}  This direction aims at choosing a committee of $k$ alternatives that ``collectively'' beats each alternative not in the committee by a majority of agents~\cite{Elkind2015:Condorcet,Alon06:Dominating}. A committee $W$ of $k$ alternatives is called a Condorcet-winning committee if, for every other alternative $b$, at least half of the agents prefer some alternative in $W$ to $b$~\cite{Elkind2015:Condorcet}. A Condorcet-winning committee $W$ is called a \emph{Condorcet-dominating committee} if, for each $b$, a majority of agents prefer a common alternative in $W$ to $b$~\cite{Alon06:Dominating,Elkind2015:Condorcet}.
\item \textbf{Direction 3: Relax the worst-case analysis.} This direction can be viewed as a (probabilistic) generalization of domain restriction that studies the probability of the (non-)existence of a Condorcet winner under distributions over agents' preferences, such as i.i.d.\@ uniform distributions over all rankings, known as the \emph{Impartial Culture} (IC). This direction was initiated by Gehrlein and Fishburn~\cite{Gehrlein1976:The-probability}, and since then, the general study of the probabilistic satisfaction of axioms, including the Condorcet criterion, has become a ``new sub-domain of the theory of social choice''~\cite{Diss2021:Evaluating}. 
\end{itemize}
It is natural to combine some of the directions mentioned above. For example, at the intersection of Directions 1 and 2, the works of Charikar, Ramakrishnan, and Wang~\cite{Charikar2026:Approximately} and Bourneuf, Charbit, and Thomass\'{e}~\cite{Bourneuf2025:Neighborhood} studied the existence of $\alpha$-dominating committees of $k$ alternatives. At the intersection of Directions 1 and 3, Caizergues, Durand, Noy, de Panafieu, and Ravelomanana~\cite{Caizergues2025:Probability} applied analytical combinatorics to study the existence of an extension of $\alpha$-Condorcet winners under IC for fixed $m$ as $n$ approaches infinity. At the intersection of all three directions, Elkind, Lang, and Saffidine~\cite[Corollary 2]{Elkind2015:Condorcet} showed that for every fixed $m$, every $k$-committee $W$, and every $\alpha<\frac{k}{k+1}$, the probability of $W$ being $\alpha$-winning under IC approaches $1$ as $n$ approaches infinity. However, the following research question, which spans all three directions, remains largely open:
\begin{quote}
How likely does an $\alpha$-winning/dominating committee of $k$ alternatives exist under IC?
\end{quote}
Much of the previous work adopted IC for a fixed number of alternatives $m$ and let the number of agents $n$ approach infinity. This setting is well-motivated in classical applications of social choice, especially in political elections, where $m$ is often relatively small and $n$ is comparatively much larger. In many novel applications of social choice, such as recommender systems~\cite{Dwork01:Rank}, combinatorial voting~\cite{Lang16:Voting} and reinforcement learning from human feedback (RLHF)~\cite{Conitzer2024:Social}, the parameter $n$ is relatively small while $m$ is much larger. Motivated by these novel applications, we focus on asymptotically addressing the aforementioned question for fixed $n$ as $m\rightarrow\infty$. That is, we will characterize
\[
\lim_{m \to \infty}\Pr_{\mathrm{IC}(n,m)}\left[\text{An $\alpha$-winning (respectively, $\alpha$-dominating) committee exists} \right].
\]

\subsection{Our Results}
Our first result reveals a phase-transition phenomenon at threshold $\frac{k-1}{k}$ for $\alpha$-winning committees.
\begin{theorem}[Proved in \cref{sec:winning}] For all constants $k \in \mathbb{N}^+$ and $\alpha \in [0, 1]$, there exists a number $N\in\mathbb{N}$ such that for every $n > N$, it holds that
\[
\lim_{m \to \infty}\Pr_{\mathrm{IC}(n,m)}\left[\text{An $\alpha$-winning committee exists} \right] = \begin{cases}
1&\text{if }\alpha< \frac{k-1}{k},\\[3pt]
0&\text{if }\alpha> \frac{k-1}{k}.
\end{cases}
\]
\end{theorem}
When $\alpha$ is below the threshold, our result is positive: an $\alpha$-winning committee exists with high probability for large $m$. When $\alpha$ is above the threshold, our result can be viewed as a probabilistic extension of the Condorcet paradox for large $m$. In the special case of $k=1$, for every constant $\alpha > 0$, an $\alpha$-winning committee is unlikely to exist when $\alpha n$ is greater than $1$ and $m$ is sufficiently large (\cref{lem:win2}).

Our second result reveals a phase transition phenomenon at threshold  $\frac{k-1}{2k}$ for $\alpha$-dominating committees. 
\begin{theorem}[Proved in \cref{sec:dominating}]
For all constants $k \in \mathbb{N}^+$ and $\alpha \in [0, 1]$, there exists a number $N\in\mathbb{N}$ such that for every $n > N$, it holds that
\[
\lim_{m \to \infty}\Pr_{\mathrm{IC}(n,m)}\left[\text{An $\alpha$-dominating committee exists} \right] =  \begin{cases}
1&\text{if }\alpha< \frac{k-1}{2k},\\[3pt]
0&\text{if }\alpha> \frac{k-1}{2k}.
\end{cases}
\]
\end{theorem}
Similarly, our result is positive when $\alpha$ is below the threshold and is an extension of Condorcet's paradox when $\alpha$ is above the threshold. Additionally, in the style of basic probabilistic methods, a straightforward corollary of the $\alpha>\frac{k-1}{2k}$ case deserves attention. Notice that in this case, for certain ranges of $n$ and $m$, there exists a preference profile (a collection of agents' preferences) under which no $\alpha$-dominating committee exists (because otherwise, the probability of its existence would be $1$). This observation is formally stated in the following corollary.
\begin{corollary}
\label{cor:dom_worst_case}
For all constants $\alpha \in [0, 1]$ and $k \in \mathbb{N}^+$ with $\alpha > \frac{k - 1}{2k}$, there exists a preference profile in which no $\alpha$-dominating $k$-committee exists.
\end{corollary}
Prior to our work, the best-known non-existence result was $\alpha>\frac{1}{3}$ for $k = 2$ (implied by the impossibility of $\alpha$-winning committees~\cite[Proposition 2]{Elkind2015:Condorcet}) and was $\alpha>\frac{1}{2} - \frac{1}{800k}$ for $k \geq 3$~\cite{Charikar2026:Approximately}. \cref{cor:dom_worst_case} improves this bound to $\alpha>\frac{1}{2} - \frac{1}{2k}$ for all $k \geq 2$.

Moreover, it has been conjectured that there always exists a $\frac{1}{3}$-winning committee of $2$ alternatives (see, e.g.,~\cite[Conjecture 1]{Charikar2025:Six-Candidates}). \cref{cor:dom_worst_case} shows that such a conjecture cannot be proved by a simple reduction to $\frac{1}{3}$-dominating $2$-committee, as these sets may not exist in general. This impossibility was not known before.

Our results are algorithmic: in each of our positive results, we show that a certain construction (which can be computed in polynomial time) of a $k$-committee is $\alpha$-winning/dominating with high probability.

\noindent\paragraph{Technical innovations.} Previous work for fixed $m$ and $n \rightarrow \infty$ primarily views the preference profile as the sum of independent (multi-dimensional) random variables---one for each agent---and then applies concentration bounds to analyze this sum. This approach does not work in our setting, as we assume the parameter $n$ to be fixed. Instead, we take the natural equivalent view~\cite{Sauermann2022:On-the-probability} of IC as generating an independent score uniformly from $[0, 1]$ for each agent--alternative pair; the agents then rank the alternatives according to their scores. Our proofs are based on carefully analyzing score distributions and identifying their underlying structures that admit or forbid the existence of $\alpha$-winning or dominating committees. Our ``large $m$ setting'' is much more technically demanding than the ``large $n$ setting,'' and we need to develop different approaches for our two main theorems.

The proofs of both sides of \cref{thm:alpha-winning} hinge on a crucial observation (\cref{lem:win2}): when $m$ is large, with high probability, there does not exist an alternative $a$ such that for every other alternative $b$, at least two agents prefer $a$ to $b$. Once we have established this observation, the negative side (when $\alpha > \frac{k - 1}{k}$) follows from a rather short and clean argument, and the positive side (when $\alpha < \frac{k - 1}{k}$) is given by a committee formed by the ``best'' alternatives for different agent groups.

The proofs for both sides of \cref{thm:alpha-dominating}, in our view, are more challenging. For the positive side (when $\alpha < \frac{k - 1}{2k}$), the crucial step is to identify a subset of alternatives that are both good and complementary in a cyclic manner (\cref{def:cyclic}). It turns out that such a subset (with proper parameters) is $\alpha$-dominating with high probability. For the negative side (when $\alpha > \frac{k - 1}{2k}$), to prove that an $\alpha$-dominating committee is unlikely to exist, we use duality arguments to show that every committee of size $k$ is blocked by a hypothetical ``mixed'' alternative. We then perform independent rounding on the ``mixed'' alternative and prove, via concentration bounds, that (1) the rounding result can, with high probability, block the committee as well, and (2) a real alternative in the election behaves similarly to the rounding result.
 
\subsection{Other Related Work}

\paragraph{The Condorcet criterion.} The Condorcet criterion is a well-studied principle that has ``nearly universal acceptance''~\cite[Page 46]{Saari1995:Basic}. The study of understanding and circumventing Condorcet's paradox has shaped the field of social choice theory. For example, Arrow's impossibility theorem~\cite{Arrow63:Social} is viewed as an extension to Condorcet's paradox---the intransitivity in social preferences is inevitable under all combinations of pairwise rules (formally, rules satisfying \emph{independence of irrelevant alternatives}) that are Pareto efficient and non-dictatorial. Another classical and important direction is domain restriction, i.e., restricting agents' preferences. For example, when agents' preferences are single-peaked, the Condorcet winner always exists~\cite{Black48:Rationale}. Sen~\cite{Sen1966:A-Possibility} identified a domain restriction for the existence of the Condorcet winner, and such domains are called \emph{Condorcet domains}. See the survey of Elkind, Lackner, and Peters~\cite{Elkind2016:Preference} for some recent developments.

\paragraph{Existence of Condorcet committees: Worst-case results.}  Elkind, Lang, and Saffidine~\cite{Elkind2015:Condorcet} proved the existence of a Condorcet-winning and Condorcet-dominating $(\lceil \log_2 m \rceil + 1)$-committee and conjectured the existence of a Condorcet-winning $3$-committee, which remains an open problem to this day. The work of Jiang, Munagala, and Wang~\cite{Jiang2020:Approximately} on approximate core-stability in committee selection implies the existence of a Condorcet-winning $32$-committee. Charikar, Lassota, Ramakrishnan, Vetta, and Wang~\cite{Charikar2025:Six-Candidates} proved the existence of an $\alpha$-winning $k$-committee when $\frac{1-\alpha}{1-\ln(1-\alpha)}\ge\frac{2}{k+1}$ (among some other parameter ranges), which implies that a Condorcet-winning $6$-committee always exists.
Nguyen, Song, and Lin~\cite{Nguyen2026:A-few-good} used novel constructions of market equilibria to improve this bound to $\alpha \cdot k^{\frac{k}{k-1}} \leq k - 1$ for all $k \ge 2$, which guarantees the existence of a Condorcet-winning $5$-committee. Independent works of Charikar, Ramakrishnan, and Wang~\cite{Charikar2026:Approximately} and Bourneuf, Charbit, and Thomass\'{e}~\cite{Bourneuf2025:Neighborhood} proved the existence of $\alpha$-dominating $k$-committees when $\alpha<\frac{1}{2}$ and $k=O\left(\left(\frac{1}{2}-\alpha\right)^{-2}\right)$.

\renewcommand{\arraystretch}{2}
\begin{table}[t]
\centering
\begin{tabular}{|c|c|c|}
\hline
\bf Model &\bf  $\alpha$-winning threshold &\bf  $\alpha$-dominating threshold\\
\hline \hline
Worst case & 
\begin{tabular}{@{}c@{}}
$\left[1 - \frac{5.10}{k}, 1 - \frac{2}{k + 1}\right]$\\
\textcolor{gray}{\footnotesize (\cite{SongN26}, \cite[Proposition 2]{Elkind2015:Condorcet})} 
\end{tabular}
& 
\begin{tabular}{@{}c@{}}
$\left[\frac{1}{2} - \Theta\bigl(\frac{1}{\sqrt{k}}\bigr), \frac{1}{2} - \frac{1}{2k}\right]$\\
\textcolor{gray}{\footnotesize (\cite{Charikar2026:Approximately,Bourneuf2025:Neighborhood}, \cref{cor:dom_worst_case})} \\
\end{tabular}
\\
\hline
IC ($n\rightarrow\infty$) &  $1 - \frac{1}{k + 1}$ \textcolor{gray}{\footnotesize (\cite[Corollary 2]{Elkind2015:Condorcet}, \cref{prop:large_n})} &  $\frac{1}{2}$ \textcolor{gray}{\footnotesize (\cref{prop:large_n})} \\
\hline
IC ($m\rightarrow\infty$) & $1 - \frac{1}{k}$ \textcolor{gray}{\footnotesize (\cref{thm:alpha-winning})} & $\frac{1}{2} - \frac{1}{2k}$ \textcolor{gray}{\footnotesize (\cref{thm:alpha-dominating})} \\
\hline
\end{tabular}
\caption{\small Existence thresholds for $\alpha$-winning and $\alpha$-dominating $k$-committees. When $\alpha$ is below the threshold, a committee exists with certainty (worst case) or with overwhelming probability (IC); when $\alpha$ is above the threshold, such committees are non-existent with certainty (worst case) or with overwhelming probability (IC).}
\label{tab:thresholds}
\end{table}
\renewcommand{\arraystretch}{1}

\paragraph{Existence of Condorcet committees: Beyond the worst case.} There is a large body of literature on the asymptotic probability of the (non-)existence of Condorcet winners for fixed $m$ under i.i.d.\@ distributions, especially in impartial cultures~\cite{Niemi1968:A-mathematical,DeMeyer1970:The-Probability,Gehrlein1976:The-probability,Tsetlin2003:The-impartial,Jones1995:Condorcet,Krishnamoorthy2005:Condorcet,Green-Armytage2016:Statistical,Brandt2016:Analyzing,Brandt2019:Exploring} and under semi-random models~\cite{Xia2020:The-Smoothed,Xia2021:Semi-Random}. Progress over the years can be found in the books by Gehrlein and Lepelley~\cite{Gehrlein2011:Voting} and by Diss and Merlin~\cite{Diss2021:Evaluating}. Elkind, Lang, and Saffidine~\cite{Elkind2015:Condorcet} proved that a Condorcet-winning $2$-committee exists with high probability under IC for fixed $m$ as $n\rightarrow\infty$. A few works studied the asymptotic probability for a Condorcet winner to exist under IC for fixed $n$ as $m$ goes to infinity. May~\cite{May1971:Some} proved (among other results) that for every fixed $n\ge 3$, the probability converges to $0$ as $m\rightarrow \infty$. Sauermann~\cite{Sauermann2022:On-the-probability} refined the convergence rate to $\Theta\left(m^{\frac{n-1}{n+1}}\right)$ for odd $n$, along with its tight leading constant and the next higher-order term. Pittel~\cite{Pittel2025:On-Likelihood} studied the case where both $m$ and $n$ go to infinity with $n\gg m^4$, and bounded the probability of the existence of a Condorcet winner by $O\left(m^{-\ell}+m^2/\sqrt n\right)$ for every $\ell > 0$. We are not aware of any previous work on the setting of $\alpha$-winning or $\alpha$-dominating $k$-committees for $k > 1$ under IC with fixed $n$ and $m \rightarrow \infty$, which we address in our work and reveal its phase-transition phenomena. \cref{tab:thresholds} summarizes our results and some of the most relevant previous work.

While IC may differ from the distributions in practice (so do other models, as ``all models are wrong''~\cite{Box1979:Robustness}), we believe that the insights revealed by our work (phase-transition phenomena in \cref{thm:alpha-winning,thm:alpha-dominating}), the implications of our work for worst-case analysis (\cref{cor:dom_worst_case}), and our proof techniques can be helpful for future studies on this promising topic motivated by novel social choice applications.
 
\section{Preliminaries}
Throughout the paper, we use the notation $[n]$ to denote the set $\{1, 2, \ldots, n\}$.

\paragraph{Profile.}
Consider an election with $m$ candidates, labeled $1,2,\dots,m$, and $n$ voters, labeled $1,2,\dots,n$.
Each voter $v \in [n]$ has a strict linear preference order (ranking) $R_v$ over $[m]$.
A \emph{profile} (with $n$ voters and $m$ candidates) is the collection $\mathcal{P} = (R_v)_{v\in [n]}$.
In our writing, we often use $v,v^\prime$ to refer to voters in $[n]$ and $a,b,c$ to refer to candidates in $[m]$.

Given a profile $\mathcal{P} = (R_v)_{v\in [n]}$ with $n$ voters and $m$ candidates, we write $a \succ_{R_v} b$ if the voter $v$ strictly (that is, $a \neq b$) prefers the candidate $a$ over the candidate $b$.
For a committee $S \subseteq [m]$, if $v$ prefers at least one candidate in $S$ over $a$, we write $S \succ_{R_v} a$; conversely, if $v$ prefers $a$ over every candidate in $S$, we write $a \succ_{R_v} S$.
We use 
\[
a\succ_{\mathcal{P}} b \coloneqq \{v\in [n]:a\succ_{R_v} b\}
\]
to denote the set of voters in $[n]$ who prefer $a$ over $b$.
Similarly, we use 
\[
S\succ_{\mathcal{P}} a \coloneqq \{v\in [n]:S\succ_{R_v} a\}
\]
to denote the set of voters in $[n]$ who prefer at least one candidate in $S$ over $a$ and
\[
a\succ_{\mathcal{P}} S \coloneqq \{v\in [n]:a\succ_{R_v} S\}
\]
to denote the set of voters in $[n]$ who prefer $a$ over every candidate in $S$.

In this paper, we consider two existing notions of collective dominance.
\begin{definition}[$\alpha$-winning]
\label{def:winning}
Given a profile $\mathcal{P}$ with parameters $n$ and $m$, a committee $S\subseteq [m]$ is \emph{$\alpha$-winning} if for all $a\in [m]\setminus S$, we have $|S\succ_{\mathcal{P}} a|\geq \alpha n$.
\end{definition}

\begin{definition}[$\alpha$-dominating]
\label{def:dominating}
Given a profile $\mathcal{P}$ with parameters $n$ and $m$, a committee $S\subseteq [m]$ is \emph{$\alpha$-dominating} if for all $a\in [m]\setminus S$, there exists $b\in S$ such that $|b\succ_\mathcal{P} a|\geq \alpha n$.
\end{definition}

Note that every $\alpha$-dominating set is also an $\alpha$-winning set.

\paragraph{Impartial Culture.}
Given parameters $n$ and $m$, the \emph{impartial culture} $\mathrm{IC}(n, m)$ generates a profile $\mathcal{P} = (R_v)_{v\in [n]}$ with $n$ voters and $m$ candidates in the following way:
Let $\mathrm{Perm}(m)$ denote the set of all $m!$ possible linear orders over $[m]$. Each voter's ranking $R_v$ is drawn
independently and uniformly at random from $\mathrm{Perm}(m)$.
In other words, voters’ preferences are independent and identically distributed (i.i.d.\@) according to the uniform distribution over all possible rankings.

\begin{definition}[$\alpha$-winning probability function]
\label{def:f}
Given an impartial culture with $n$ voters and $m$ candidates, let the \emph{$\alpha$-winning probability function} $f_{\alpha,k}(n,m)$ denote the probability that there exists an $\alpha$-winning committee $S \subseteq [m]$ of size $k$, that is,
\[
f_{\alpha,k}(n,m)
\coloneqq
\Pr_{\mathcal{P} \sim \mathrm{IC}(n,m)}\left[
  \exists S \subseteq [m] \text{ s.t.\@ } |S| = k \text{ and }
  S \text{ is } \alpha\text{-winning in } \mathcal{P}
\right].
\]
\end{definition}

\begin{definition}[$\alpha$-dominating probability function]
\label{def:g}
Given an impartial culture with $n$ voters and $m$ candidates, 
let the \emph{$\alpha$-dominating probability function} $g_{\alpha,k}(n,m)$ denote the probability that there exists an $\alpha$-dominating committee $S \subseteq [m]$ of size $k$, that is,
\[
g_{\alpha,k}(n,m)
\coloneqq
\Pr_{\mathcal{P} \sim \mathrm{IC}(n,m)}\left[
  \exists S \subseteq [m] \text{ s.t.\@ } |S| = k \text{ and }
  S \text{ is } \alpha\text{-dominating in } \mathcal{P}
\right].
\]
\end{definition}

\paragraph{Score.}
Here is an equivalent method to generate the profile $\mathcal{P}=(R_v)_{v\in [n]}$ in an impartial culture with $n$ voters and $m$ candidates: Each voter $v$ independently samples a score $s_{v,a}$ for each candidate $a$ uniformly at random from $[0,1]$. (Note that ties occur with probability $0$.) The voter $v$ ranks the candidates in decreasing order of these scores, thus obtaining $v$'s ranking $R_v$.
In this way, 
\[
a\succ_{R_v} b \quad\iff\quad s_{v,a}>s_{v,b}
\]
and
\[
S\succ_{R_v} a \quad\iff\quad \max_{c\in S}s_{v,c}>s_{v,a}.
\]
We define the score matrix $\mathbf{S}\coloneqq (s_{v,a})_{v\in [n], a\in [m]}$, where $s_{v,a}\sim \mathrm{Unif}[0,1]$.

\paragraph{Asymptotic Notations.}
Let \(F(n, m)\) and \(G(n, m)\) be non-negative functions in \(n\) and \(m\). We use the following notations:
\begin{itemize}
  \item \(F(n,m) = O_n(G(n,m))\) if there exists a function \(C(n)\) that depends only on \(n\) such that for all sufficiently large \(m\),
  \[
    F(n,m) \le C(n)\, G(n,m).
  \]

  \item \(F(n,m) = \widetilde{O}_n(G(n,m))\) if there exist functions \(C(n)\) and \(p(n) \ge 0\) that depend only on \(n\) such that for all sufficiently large \(m\),
  \[
    F(n,m) \le C(n)\, (\ln m)^{p(n)}\, G(n,m).
  \]

  \item \(F(n,m) = o_n(G(n,m))\) if there exists a function \(C(n)\) that depends only on \(n\) such that for every constant \(\varepsilon>0\) there exists a constant \(M(\varepsilon)\) such that for all \(m\ge M(\varepsilon)\),
  \[
    F(n,m) \le \varepsilon\, C(n)\, G(n,m).
  \]
\end{itemize}
 
\section{\texorpdfstring{$\alpha$}{Alpha}-Winning Sets in the Limit}
\label{sec:winning}
In this section, we study the probability that there exists an $\alpha$-winning set (\cref{def:winning}) of size $k$ in a large impartial culture in which the number of candidates, $m$, is much greater than the number of voters, $n$. More precisely, we will prove \cref{thm:alpha-winning}. Recall that $f_{\alpha,k}(n,m)$ denotes the probability that there exists an $\alpha$-winning committee $S \subseteq [m]$ of size $k$ in an impartial culture with $n$ voters and $m$ candidates.
\begin{theorem}
\label{thm:alpha-winning}
For all constants $k \in \mathbb{N}^+$ and $\alpha \in [0, 1]$, there exists a number $N\in\mathbb{N}$ such that for every $n > N$, it holds that
\[
\lim_{m \to \infty} f_{\alpha,k}(n,m) =
\begin{cases}
1&\text{if }\alpha< \frac{k-1}{k},\\[3pt]
0&\text{if }\alpha> \frac{k-1}{k}.
\end{cases}
\]
\end{theorem}
To prove the theorem, we first establish the following crucial lemma, which states that, with high probability, there is no candidate that wins at least $2$ voters against every other candidate when $m$ is much greater than $n$.

\begin{lemma}
\label{lem:win2}
For all $n\in \mathbb{N}^{+}$, it holds that
\[
\lim\limits_{m \to \infty}
\Pr\limits_{\mathcal{P} \sim \mathrm{IC}(n,m)} \left[
\exists a \in [m] \text{ s.t.\@ }
\forall b \in [m] \setminus \{a\},\;
|a \succ_{\mathcal{P}} b| \ge 2
\right]
= 0.
\]
\end{lemma}
\begin{proof}
We first apply the union bound on the probability that such a candidate exists:
\begin{align*}
&\ \Pr_{\mathcal{P} \sim \mathrm{IC}(n,m)}\!\left[
\exists a\in [m] \text{ s.t.\@ }
\forall b\in [m]\setminus\{a\},\; |a \succ_{\mathcal{P}} b|\ge 2
\right]\\
\le &\sum_{a\in [m]}
\Pr_{\mathcal{P}\sim \mathrm{IC}(n,m)}\!\left[
\forall b\in [m]\setminus\{a\},\; |a \succ_{\mathcal{P}} b|\ge 2
\right].
\end{align*}

By symmetry (each candidate has the same distribution under IC), the sum on the right-hand side is equal to $m$ times the probability for a fixed candidate. Hence, for any fixed $a_0\in [m]$, we have
\begin{align*}
&\ \Pr_{\mathcal{P} \sim \mathrm{IC}(n,m)}\!\left[
\exists a\in [m] \text{ s.t.\@ }
\forall b\in [m]\setminus\{a\},\; |a \succ_{\mathcal{P}} b|\ge 2
\right]\\
\leq&\ m \cdot 
\Pr_{\mathcal{P} \sim \mathrm{IC}(n,m)}\!\left[
\forall b\in [m]\setminus\{a_0\},\; |a_0 \succ_{\mathcal{P}} b|\ge 2
\right].
\end{align*}

We now analyze the probability that the candidate $a_0$ pairwise wins at least $2$ voters against every other candidate. 
In order to prove \cref{lem:win2}, it suffices to show the following:
\[
\Pr_{\mathcal{P} \sim \mathrm{IC}(n,m)}\!\left[
\forall b\in [m]\setminus\{a_0\},\; |a_0 \succ_{\mathcal{P}} b|\ge 2
\right]=o_n(m^{-1}).
\]

We can reduce the probability calculation on the preference profile $\mathcal{P}$ to an equivalent calculation on the score matrix $\mathbf{S}$.
For simplicity, let us write $\mathbf{s} = (s_{v,a_0})_{v \in [n]}$ and $s_v=s_{v,a_0}$. 
Let $\Pr[\cdot\mid\mathbf{s}]$ denote the conditional probability taken over the randomness of the other columns $\{(s_{v,b})_{v\in [n]}:b\neq a\}$.
Then, we have
\begin{align*}
&\ \Pr_{\mathcal{P} \sim \mathrm{IC}(n,m)}\!\left[
\forall b\in [m]\setminus\{a_0\},\; |a_0 \succ_{\mathcal{P}} b|\ge 2
\right]\\
=&\ \Pr_{\mathbf{S}\sim\mathrm{Unif}[0,1]^{n\times m}}\!\left[
\forall b\in [m]\setminus\{a_0\},\; |a_0 \succ_{\mathcal{P}} b|\ge 2
\right]\\
=&\ \int_{[0,1]^n}
\Pr\left[\forall b\in [m]\setminus\{a_0\},\; |a_0\succ_{\mathcal{P}} b|\ge 2
\mid \mathbf{s}\right]
\,\d\mathbf{s}.
\end{align*}
For all $b\neq a_0$, the columns $(s_{v,b})_{v\in [n]}$ have independent and identical distributions. Therefore, for any fixed $b_0\in[m]\setminus\{a_0\}$, we have
\begin{align*}
&\ \Pr\left[\forall b\in [m]\setminus\{a_0\},\; |a_0\succ_{\mathcal{P}} b|\ge 2
\mid \mathbf{s}\right]\\
=&\ \prod_{b\in [m]\setminus\{a_0\}}
\Pr\bigl[\,|a_0\succ_{\mathcal{P}} b|\ge 2\mid \mathbf{s}\bigr]\\
=&\ \left(\Pr\left[
|a_0 \succ_{\mathcal{P}} b_0|\ge 2
\mid\mathbf{s}\right]\right)^{m-1}.
\end{align*}

Fix $a_0$'s scores $\mathbf{s} = (s_{v,a_0})_{v \in [n]}$.
The probability that $a_0$ wins no voter against $b_0$ is 
\[
\prod_{v\in [n]}(1-s_v),
\]
and the probability that $a_0$ wins exactly $1$ voter against $b_0$ is 
\[
\sum_{v\in [n]}s_v\prod_{v^\prime\in [n]\setminus\{v\}}(1-s_{v^\prime}).
\]
Hence, the probability that $a_0$ wins at least $2$ voters against $b_0$ is
\begin{align*}
&\ \int_{[0,1]^n}\left(\Pr\left[
|a_0 \succ_{\mathcal{P}} b_0| \ge 2\mid\mathbf{s}\right]\right)^{m-1}\d\mathbf{s}\\
=&\ \int_{[0,1]^n}\left(
1-\prod_{v\in [n]}(1-s_v)-\sum_{v\in [n]}s_v\prod_{v^\prime\in [n]\setminus\{v\}}(1-s_{v^\prime})
\right)^{m-1}\d\mathbf{s}.
\end{align*}
We rewrite
\[
1-\prod_{v\in [n]}(1-s_v)-\sum_{v\in [n]}s_v\prod_{v^\prime\in [n]\setminus\{v\}}(1-s_{v^\prime})
=1-\left(1+\sum_{v\in [n]}\frac{s_v}{1-s_v}\right)\cdot\prod_{v\in [n]}(1-s_v).
\]
Since
\[
1+\sum_{v\in [n]}\frac{s_v}{1-s_v}\ge 1+\max_{v\in [n]}\frac{s_v}{1-s_v}=\frac{1}{\min_{v\in [n]}(1-s_v)},
\]
we obtain the upper bound
\[
\int_{[0,1]^n}\left(\Pr\left[
|a_0 \succ_{\mathcal{P}} b_0| \ge 2\mid\mathbf{s}\right]\right)^{m-1}\d\mathbf{s}
\le \int_{[0,1]^n}\left(1-\frac{\prod_{v\in [n]}(1-s_v)}{\min_{v\in [n]}(1-s_v)}\right)^{m-1}\d\mathbf{s}.
\]
Since
\[
\min_{v\in [n]}(1-s_v)\le\left(\prod_{v\in [n]}(1-s_v)\right)^{1/n},
\]
dividing $\prod_{v\in [n]}(1-s_v)$ by both sides yields
\[
\frac{\prod_{v\in [n]}(1-s_v)}{\min_{v\in [n]}(1-s_v)}
\ge\left(\prod_{v\in [n]}(1-s_v)\right)^{(n-1)/n}.
\]

For simplicity, define \[
B(\mathbf{s}) \coloneqq \left(\prod_{v\in V}(1-s_v)\right)^{(n-1)/n},
\]
and let
\[
D \coloneqq \left\{\mathbf{s}\in[0,1]^n\mid B(\mathbf{s})\le\frac{2\ln m}{m-1}\right\}.
\]
Then we have
\begin{align}
&\ \int_{[0,1]^n}\left(1-\frac{\prod_{v\in [n]}(1-s_v)}{\min_{v\in [n]}(1-s_v)}\right)^{m-1}\d\mathbf{s} \nonumber\\
\le&\ \int_{[0,1]^n}\left(1-B(\mathbf{s})\right)^{m-1}\d\mathbf{s} \nonumber\\
=&\ \int_{D}\left(1-B(\mathbf{s})\right)^{m-1}\d\mathbf{s}
+\int_{[0,1]^n\setminus D}\left(1-B(\mathbf{s})\right)^{m-1}\d\mathbf{s}. \label{eq:two_terms}
\end{align}

For any $\mathbf{s}\in [0,1]^n\setminus D$, using the fact that $1 + x \leq e^x$ for all $x \in \mathbb{R}$, we have
\[
\left(1-B(\mathbf{s})\right)^{m-1}\le \left(1-\frac{2\ln m}{m-1}\right)^{m-1}
\le\left(\exp\left(-\frac{2\ln m}{m-1}\right)\right)^{m-1}=m^{-2}.
\]
Hence, we can bound the second term in \cref{eq:two_terms} as follows:
\[
\int_{[0,1]^n\setminus D}\left(1-B(\mathbf{s})\right)^{m-1}\d\mathbf{s}
\le \text{Vol}([0,1]^n\setminus D)\cdot m^{-2}
\le m^{-2}.
\]

For the first term in \cref{eq:two_terms}, we have
\[
\int_{D}\left(1-B(\mathbf{s})\right)^{m-1}\d\mathbf{s}
\le \text{Vol}(D)\cdot 1=\text{Vol}(D)=\Pr_{\mathbf{s}\sim\mathrm{Unif}[0,1]^n}\left[B(\mathbf{s})\le\frac{2\ln m}{m-1}\right].
\]
We will then bound $\Pr_{\mathbf{s}\sim\mathrm{Unif}[0,1]^n}\left[B(\mathbf{s})\le\frac{2\ln m}{m-1}\right]$. We first have
\begin{align*}
&\ \Pr_{\mathbf{s}\sim\mathrm{Unif}[0,1]^n}\left[B(\mathbf{s})\le\frac{2\ln m}{m-1}\right]\\
=&\ \Pr_{\mathbf{s}\sim\mathrm{Unif}[0,1]^n}\left[\prod_{v\in [n]}(1-s_v)\le\left(\frac{2\ln m}{m-1}\right)^{n/(n-1)}\right]\\
=&\ \Pr_{\mathbf{s}\sim\mathrm{Unif}[0,1]^n}\left[\sum_{v\in [n]}-\ln(1-s_v)\ge-\frac{n}{n-1}\ln\left(\frac{2\ln m}{m-1}\right)\right].
\end{align*}

We calculate the cumulative distribution function (CDF) of $x_v=-\ln(1-s_v)$ as follows. For any $z \in [0, +\infty)$, we have
\[
\Pr_{s_v\sim\mathrm{Unif}[0,1]}[x_v > z]=\Pr_{s_v\sim\mathrm{Unif}[0,1]}[-\ln(1-s_v) > z]=\Pr_{s_v\sim\mathrm{Unif}[0,1]}[1-s_v < e^{-z}]=e^{-z}.
\]
Therefore, the CDF of $x_v$ is $1-e^{-z}$, which means $x_v$ follows an exponential distribution, and thus $\sum_{v\in [n]}x_v$ follows a Gamma distribution: for all $z \geq 0$, we have
\[
\Pr\left[\sum_{v\in[n]}x_v < z\right]=1-e^{-z}\sum_{t=0}^{n-1}\frac{z^t}{t!}.
\]
We therefore obtain
\begin{align*}
&\ \Pr_{\mathbf{s}\sim\mathrm{Unif}[0,1]^n}\left[B(\mathbf{s})\le\frac{2\ln m}{m-1}\right]\\
= &\ \Pr
\left[\sum_{v\in [n]}x_v\ge-\frac{n}{n-1}\ln\left(\frac{2\ln m}{m-1}\right)\right]\\
= &\ \left(\frac{2\ln m}{m-1}\right)^{\frac{n}{n-1}}\sum_{t=0}^{n-1}\frac{\left(-\frac{n}{n-1}\ln\left(\frac{2\ln m}{m-1}\right)\right)^t}{t!}\\
= &\ O_n\left(\left(\frac{2\ln m}{m-1}\right)^{\frac{n}{n-1}}\frac{\left(-\frac{n}{n-1}\ln\left(\frac{2\ln m}{m-1}\right)\right)^{n-1}}{(n-1)!}\right)\\
= &\ \widetilde{O}_n\left(m^{-\frac{n}{n-1}}\right).
\end{align*}

Finally, we have
\begin{align*}
&\ \Pr_{\mathcal{P} \sim \mathrm{IC}(n,m)}\!\left[
\forall b\in [m]\setminus\{a_0\},\; |a_0 \succ_{\mathcal{P}} b|\ge 2
\right]\\
\le&\ \int_{D}\left(1-B(\mathbf{s})\right)^{m-1}\d\mathbf{s}
+\int_{[0,1]^n\setminus D}\left(1-B(\mathbf{s})\right)^{m-1}\d\mathbf{s}\\
\le&\ \widetilde{O}_n\left(m^{-\frac{n}{n-1}}\right)+m^{-2}\\
=&\ o_n\left(m^{-1}\right).\qedhere
\end{align*}
\end{proof}

With \cref{lem:win2}, we can get the threshold of $\alpha$ when $k=1$.

\begin{corollary}
\label{cor:alpha-winning1}
For any constants $\alpha\in(0,1]$ and $n > \frac{1}{\alpha}$, it holds that
\[
\lim_{m\to\infty}f_{\alpha,1}(n,m)=0,
\]
\end{corollary}
\begin{proof}[Proof]
According to the definition of $f_{\alpha,1}(n,m)$, we have
\begin{align*}
&\ \lim_{m\to\infty}f_{\alpha,1}(n,m)\\
=& \lim_{m\to\infty}
\Pr_{\mathcal{P} \sim \mathrm{IC}(n,m)}\left[
  \exists S \subseteq [m]\text{ s.t. } |S| = 1 \text{ and } 
  S \text{ is } \alpha\text{-winning in } \mathcal{P}
\right]\\
=& \lim_{m\to\infty}
\Pr_{\mathcal{P} \sim \mathrm{IC}(n,m)}\left[
  \exists a \subseteq [m]\text{ s.t. } 
  \forall b \in [m] \setminus \{a\},\;
|a \succ_{\mathcal{P}} b| \ge \alpha n
\right].
\end{align*}
For any constants $\alpha \in (0, 1]$ and $n$ with $\alpha n > 1$, by \cref{lem:win2}, we have
\begin{align*}
& \lim_{m\to\infty}
\Pr_{\mathcal{P} \sim \mathrm{IC}(n,m)}\left[
  \exists a \subseteq [m]\text{ s.t. } 
  \forall b \in [m] \setminus \{a\},\;
|a \succ_{\mathcal{P}} b| \ge \alpha n
\right]\\
\le&\ \lim_{m\to\infty}
\Pr_{\mathcal{P} \sim \mathrm{IC}(n,m)}\left[
  \exists a \subseteq [m]\text{ s.t. } 
  \forall b \in [m] \setminus \{a\},\;
|a \succ_{\mathcal{P}} b| \ge 2
\right]\\
=&\ 0.\qedhere
\end{align*}
\end{proof}

This corollary shows that, for any $\alpha \in (0,1]$, with high probability, there is no $\alpha$-winning candidate (a committee of size $1$) when $n$ is large enough and $m$ is much greater than $n$.

Now that everything is ready, we return to the proof of the theorem. We have already handled the scenario of $k=1$, so it remains to consider the scenario of $k\ge 2$. The argument naturally splits into two cases depending on the value of $\alpha$.

\paragraph{Case 1 (positive side): $\alpha \in \left[0,\frac{k-1}{k}\right)$.}
In this case, we show that for a sufficiently large $n$,
\[
\lim_{m\to\infty} f_{\alpha,k}(n,m)=1.
\]
We partition the voters into voter groups $V_1, V_2, \ldots, V_k$, where
\[
V_i \coloneqq \left\{\left\lceil\frac{i-1}{k}n\right\rceil+1,\dots,\left\lceil\frac{i}{k}n\right\rceil\right\}
\]
For each $i \in [k]$, we define the ``best'' candidate in $V_i$ as
\[
c_i \coloneqq \arg\max_a \sum_{v\in V_i} s_{a,v}.
\]
We select a committee $S=\{c_i\mid i\in[k]\}$. We will prove that $S$ is probably $\alpha$-winning for $\alpha\in\left[0,\frac{k-1}{k}\right)$. We have
\begin{align*}
&\ f_{\alpha,k}(n,m)\\
\ge &\ \Pr_{\mathcal{P} \sim \mathrm{IC}(n,m)}\left[ 
  S \text{ is } \alpha\text{-winning in } \mathcal{P}\right]\\
=&\ \Pr_{\mathcal{P} \sim \mathrm{IC}(n,m)}\left[ 
  \forall a\in [m]\setminus S,\ |S\succ_{\mathcal{P}} a|\ge\alpha n\right]\\
=&\ 1-\Pr_{\mathcal{P} \sim \mathrm{IC}(n,m)}\left[ 
  \exists a\in [m]\setminus S,\ |a\succ_{\mathcal{P}} S|\ge(1-\alpha) n\right].
\end{align*}
By the union bound and symmetry, we further have
\begin{align*}
&\ f_{\alpha,k}(n,m)\\
\ge &\ 1 - \sum_{a\in [m] \setminus S}\Pr_{\mathcal{P} \sim \mathrm{IC}(n,m)}\left[ 
  |a\succ_{\mathcal{P}} S|\ge(1-\alpha) n\right]\\
\geq &\ 1 - m \cdot\Pr_{\mathcal{P} \sim \mathrm{IC}(n,m)}\left[ 
  |a_0\succ_{\mathcal{P}} S|\ge(1-\alpha) n\right],
\end{align*}
where $a_0$ is an arbitrary fixed candidate in $[m] \setminus S$.

Therefore, if we can obtain \[
\Pr_{\mathcal{P} \sim \mathrm{IC}(n,m)}\left[ 
  |a_0\succ_{\mathcal{P}} S|\ge(1-\alpha) n\right]=o_n\left(\frac{1}{m}\right)
\] when $n\to \infty$, we will finish the proof for $\alpha \in \left[0,\frac{k-1}{k}\right)$. Note that the candidate $a_0$ was not selected into $S$, and hence its score vector is first-order stochastically dominated by that of a fresh candidate. In the subsequent proof of the bound above, we relax the score vector of $a_0$ to that of a fresh candidate.

For each $i\in[k]$, define a sub-profile $\mathcal{P}_i\coloneqq(R_v)_{v\in V_i}$, and similarly define
\[
a_0\succ_{\mathcal{P}_i} S \coloneqq \{v\in V_i:a_0\succ_{R_v} S\}.
\]

We now bound the probability that $a_0$ wins at least $\lceil(1-\alpha)n\rceil$ voters against $S$. By the union bound and independence across voter groups, we have
\begin{align*}
&\ \Pr_{\mathcal{P} \sim \mathrm{IC}(n,m)}\left[ 
  |a_0\succ_{\mathcal{P}} S|\ge(1-\alpha) n\right]\\
\le&\ \sum_{\substack{q_1+\dots+q_k=\lceil(1-\alpha)n\rceil \\ q_1,\dots,q_k\in\{0,\dots,\lceil n/k\rceil\}}}\prod_{i=1}^k\Pr_{\mathcal{P} \sim \mathrm{IC}(n,m)}[|a_0\succ_{\mathcal{P}_i}S|\ge q_i].
\end{align*}
Define the event $\mathcal{A}_{i,q}\coloneqq\{|a_0\succ_{\mathcal{P}_i}S|\ge q\}$ and we write $\Pr_{\mathcal{P} \sim \mathrm{IC}(n,m)}[\mathcal{A}_{i,q}]$ as $\Pr[\mathcal{A}_{i,q}]$. Then we will bound $\Pr[\mathcal{A}_{i,q}]$.
\begin{align*}
&\ \Pr[\mathcal{A}_{i,q}]\\
=&\ \Pr[\mathcal{A}_{i,q}\mid\mathcal{A}_{i,q-1}]\cdot\Pr[\mathcal{A}_{i,q-1}]+\Pr[\mathcal{A}_{i,q}\mid\neg\mathcal{A}_{i,q-1}]\cdot\Pr[\neg\mathcal{A}_{i,q-1}]\\
=&\ \Pr[\mathcal{A}_{i,q}\mid\mathcal{A}_{i,q-1}]\cdot\Pr[\mathcal{A}_{i,q-1}]\\
=&\ \Pr[\mathcal{A}_{i,q}\mid\mathcal{A}_{i,q-1}]\cdot\Pr[\mathcal{A}_{i,q-1}\mid\mathcal{A}_{i,q-2}]\cdots \Pr[\mathcal{A}_{i,1}\mid\mathcal{A}_{i,0}]\\
=&\ \prod_{q'=1}^{q}\Pr[\mathcal{A}_{i,q'}\mid\mathcal{A}_{i,q'-1}].
\end{align*}

Then we will use the following lemma to provide a simple bound for $\prod_{q'=1}^{q}\Pr[\mathcal{A}_{i,q'}\mid\mathcal{A}_{i,q'-1}]$.
\begin{lemma}
\label{lem:A0}
For all $t\in\{1,\ldots,q\}$, we have
\[
\Pr_{\mathcal{P} \sim \mathrm{IC}(n,m)}[\mathcal{A}_{i,t}\mid\mathcal{A}_{i,t-1}]\le
\Pr_{\mathcal{P} \sim \mathrm{IC}(n,m)}[\mathcal{A}_{i,1}].
\]
\end{lemma}
\begin{proof}
Recall that $V_i=\left\{\left\lceil\frac{i-1}{k}n\right\rceil+1,\dots,\left\lceil\frac{i}{k}n\right\rceil\right\}$.
For any integer $j\in \left[\left\lceil\frac{i-1}{k}n\right\rceil+1,\left\lceil\frac{i}{k}n\right\rceil\right]$, we denote $\left\{\left\lceil\frac{i-1}{k}n\right\rceil+1,\dots,j\right\}$ as $V_i^j$.

Define the event $\mathcal{B}_{i,t}^j$ as
\[
\mathcal{B}_{i,t}^j\coloneqq \left(\left|\{v\in V_i^j:a_0\succ_{R_v} S\}\right|=t\right)
\land \left(\left|\{v\in V_i^{j-1}:a_0\succ_{R_v} S\}\right|=t-1\right).
\]
The event $\mathcal{B}_{i,t}^j$ occurs when, among the first $j$ voters in $V_i$, exactly $t$ voters rank $a_0$ above $S$, while among the first $j-1$ voters, exactly $t-1$ voters do so. Equivalently, the $j$-th voter is the one at which $a_0$ gets its $t$-th win against $S$. We have
\begin{align*}
&\ \Pr[\mathcal{A}_{i,t}\mid\mathcal{A}_{i,t-1}]\\
=&\ \sum_{j=\left\lceil\frac{i-1}{k}n\right\rceil+1}^{\left\lceil\frac{i}{k}n\right\rceil}
\Pr[\mathcal{B}_{i,t-1}^j
\land\mathcal{A}_{i,t}\mid\mathcal{A}_{i,t-1}]\tag{law of total probability}\\
=&\ \sum_{j=\left\lceil\frac{i-1}{k}n\right\rceil+1}^{\left\lceil\frac{i}{k}n\right\rceil}
\Pr[\mathcal{B}_{i,t-1}^j\mid\mathcal{A}_{i,t-1}]\cdot
\Pr[|a_0\succ_{V_i\setminus V_i^j} S|\ge 1\mid\mathcal{B}_{i,t-1}^j\land\mathcal{A}_{i,t-1}]\tag{definition of conditional probability}\\
=&\ \sum_{j=\left\lceil\frac{i-1}{k}n\right\rceil+1}^{\left\lceil\frac{i}{k}n\right\rceil}
\Pr[\mathcal{B}_{i,t-1}^j\mid\mathcal{A}_{i,t-1}]\cdot
\Pr[|a_0\succ_{V_i\setminus V_i^j} S|\ge 1]\tag{independence}\\
\le&\ \sum_{j=\left\lceil\frac{i-1}{k}n\right\rceil+1}^{\left\lceil\frac{i}{k}n\right\rceil}
\Pr[\mathcal{B}_{i,t-1}^j\mid\mathcal{A}_{i,t-1}]\cdot
\Pr[\mathcal{A}_{i,1}]\\
=&\ \Pr[\mathcal{A}_{i,1}]. \tag{law of total probability}
\end{align*}
This completes the proof of the lemma.
\end{proof}

By \cref{lem:A0}, we have
\[
\prod_{q'=1}^{q}\Pr[\mathcal{A}_{i,q'}\mid\mathcal{A}_{i,q'-1}]\le \Pr[\mathcal{A}_{i,1}]^q.
\]
Then we will bound $\Pr[\mathcal{A}_{i,1}]$.

Let $U_i\coloneqq \sum_{v\in V_i} s_{c_i,v} = \max_a \sum_{v\in V_i} s_{a,v}$. We have
\begin{align*}
&\ \Pr[\mathcal{A}_{i,1}]\\
\le&\ \Pr[\exists v\in V_i\text{ s.t.\@ }s_{c_i,v}<s_{a_0,v}]\\
\le&\ \sum_{v\in V_i}\Pr[s_{c_i,v}<s_{a_0,v}]\\
=&\ \sum_{v\in V_i}\left(1-\E[s_{c_i,v}]\right)\\
=&\ |V_i|-\E[U_i].
\end{align*}

Notice that $ \sum_{v\in V_i} s_{a,v}$ follows the Irwin--Hall distribution.
By the CDF of the Irwin--Hall distribution, for any $x\in[0,1]$, we have
\[
\Pr\left[\sum_{v\in V_i}s_{a,v}<|V_i|-x\right]=1-\frac{x^{|V_i|}}{|V_i|!}.
\]
Since all candidates are independent and $U_i = \max_a \sum_{v\in V_i} s_{a,v}$, we have
\[
\Pr[U_i<|V_i| - x]=\left(1-\frac{x^{|V_i|}}{|V_i|!}\right)^m.
\]

We choose the value for $x$ so that $\frac{x^{|V_i|}}{|V_i|!}=\frac{\ln m}{m}$; that is, $x=\left(|V_i|!\frac{\ln m}{m}\right)^{\frac{1}{|V_i|}}$. (For a sufficiently large $m$, we have $x \in [0, 1]$.) Then,
\begin{align*}
&\ |V_i|-\E[U_i]\\
\le&\ \Pr[U_i\ge |V_i|-x]\cdot x+\Pr[U_i<|V_i|-x]\cdot |V_i|\\
=&\ \left(1-\left(1-\frac{\ln m}{m}\right)^m\right)
\cdot\left(|V_i|!\cdot\frac{\ln m}{m}\right)^{\frac{1}{|V_i|}}
+\left(1-\frac{\ln m}{m}\right)^m\cdot|V_i|\\
\le&\ \left(1-\left(1-\frac{\ln m}{m}\right)^m\right)
\cdot\left(\left\lceil\frac{n}{k}\right\rceil!\cdot\frac{\ln m}{m}\right)^{\frac{1}{\left\lceil\frac{n}{k}\right\rceil}}
+\left(1-\frac{\ln m}{m}\right)^m\cdot\left\lceil\frac{n}{k}\right\rceil.
\end{align*}

Recall that our goal is to prove $\Pr_{\mathcal{P} \sim \mathrm{IC}(n,m)}\left[ 
  |a_0\succ_{\mathcal{P}} S|\ge(1-\alpha) n\right]=o_n\left(\frac{1}{m}\right)$. We have
\begin{align*}
&\Pr_{\mathcal{P} \sim \mathrm{IC}(n,m)}\left[ 
  |a_0\succ_{\mathcal{P}} S|\ge(1-\alpha) n\right]\\
\le &\ \sum_{\substack{q_1+\dots+q_k=\lceil(1-\alpha)n\rceil \\ q_1,\dots,q_k\in\{0,\dots,\lceil n/k\rceil\}}}\prod_{i=1}^k\Pr[\mathcal{A}_{i,q_i}]\\
\le&\ \sum_{\substack{q_1+\dots+q_k=\lceil(1-\alpha)n\rceil \\ q_1,\dots,q_k\in\{0,\dots,\lceil n/k\rceil\}}}\prod_{i=1}^k
\left(\left(1-\left(1-\frac{\ln m}{m}\right)^m\right)
\cdot\left(\left\lceil\frac{n}{k}\right\rceil!\cdot\frac{\ln m}{m}\right)^{\frac{1}{\left\lceil\frac{n}{k}\right\rceil}}
+\left(1-\frac{\ln m}{m}\right)^m\cdot\left\lceil\frac{n}{k}\right\rceil\right)^{q_i}.
\end{align*}
Notice that
\[
\left(1-\frac{\ln m}{m}\right)^m \le\left(\exp\left(-\frac{\ln m}{m}\right)\right)^m = \frac{1}{m}.
\]
Then we have (for a sufficiently large $m$)
\begin{align*}
&\ \left(1-\left(1-\frac{\ln m}{m}\right)^m\right)
\cdot\left(\left\lceil\frac{n}{k}\right\rceil!\cdot\frac{\ln m}{m}\right)^{\frac{1}{\left\lceil\frac{n}{k}\right\rceil}}
+\left(1-\frac{\ln m}{m}\right)^m\cdot\left\lceil\frac{n}{k}\right\rceil\\
\le&\ \left(\left\lceil\frac{n}{k}\right\rceil! \cdot \frac{\ln m}{m}\right)^{\frac{1}{\left\lceil\frac{n}{k}\right\rceil}} + \frac{1}{m}\cdot\left\lceil\frac{n}{k}\right\rceil\\
\le&\ \left(\frac{\ln m}{m}\right)^{\frac{1}{\left\lceil\frac{n}{k}\right\rceil}}\left(\left(\left\lceil\frac{n}{k}\right\rceil!\right)^{\frac{1}{\left\lceil\frac{n}{k}\right\rceil}}+1\right).
\end{align*}
Thus, we have
\begin{align*}
&\Pr_{\mathcal{P} \sim \mathrm{IC}(n,m)}\left[ 
  |a_0\succ_{\mathcal{P}} S|\ge(1-\alpha) n\right]\\
\le&\ \sum_{\substack{q_1+\dots+q_k=\lceil(1-\alpha)n\rceil \\ q_1,\dots,q_k\in\{0,\dots,\lceil n/k\rceil\}}}\prod_{i=1}^k
\left(\frac{\ln m}{m}\right)^{\frac{q_i}{\left\lceil\frac{n}{k}\right\rceil}}\left(\left(\left\lceil\frac{n}{k}\right\rceil!\right)^{\frac{1}{\left\lceil\frac{n}{k}\right\rceil}}+1\right)^{q_i}\\
\le&\ \left(\left\lceil\frac{n}{k}\right\rceil+1\right)^k\cdot
\left(\frac{\ln m}{m}\right)^{\frac{\lceil(1-\alpha)n\rceil}{\left\lceil\frac{n}{k}\right\rceil}}\cdot
\left(\left(\left\lceil\frac{n}{k}\right\rceil!\right)^{\frac{1}{\left\lceil\frac{n}{k}\right\rceil}}+1\right)^{\lceil(1-\alpha)n\rceil},
\end{align*}
where $\left(\lceil\frac{n}{k}\rceil+1\right)^k$ is an upper bound on the number of all possible $(q_1,\dots,q_k)$.

For any $\alpha\in\left[0,\frac{k-1}{k}\right)$, there is a constant $\varepsilon_\alpha > 0$, such that for all sufficiently large $n$, we have
\[
\frac{\lceil(1-\alpha)n\rceil}{\left\lceil\frac{n}{k}\right\rceil}>1 + \varepsilon_\alpha.
\]
Therefore, we obtain
\[
\Pr_{\mathcal{P} \sim \mathrm{IC}(n,m)}\left[ 
  |a_0\succ_{\mathcal{P}} S|\ge(1-\alpha) n\right]=o_n\left(\frac{1}{m}\right)
\]
as desired.

\paragraph{Case 2 (negative side): $\alpha \in \left(\frac{k-1}{k},1\right]$.}
Here we show that for a sufficiently large $n$,
\[
\lim_{m\to\infty} f_{\alpha,k}(n,m)=0.
\]

We will prove this by contradiction.

By the pigeonhole principle, for a committee $S=\{c_1,\dots,c_k\}$ of size $k$, there exists a candidate $c$ in $S$ such that
\[
\left|\left\{v\in [n]: s_{v,c} = \max_{i \in [k]} s_{v,c_i}\right\}\right|\ge \left\lceil\frac{n}{k}\right\rceil.
\]
Therefore, if $S$ is an $\alpha$-winning committee for $\alpha\in\left(\frac{k-1}{k},1\right]$, the candidate $c$ should win at least $\alpha n-\left\lfloor\frac{k-1}{k}n\right\rfloor$ voters against every other candidate. In other words, for each candidate $b\in [m]\setminus\{c\}$, it holds that $|c\succ_{\mathcal{P}} b|\ge \alpha n-\left\lfloor\frac{k-1}{k}n\right\rfloor$.

Notice that $\alpha n-\left\lfloor\frac{k-1}{k}n\right\rfloor > 1$ when $n$ is sufficiently large. Therefore, by \cref{lem:win2}, we have
\[
\lim_{m \to \infty}
\Pr_{\mathcal{P} \sim \mathrm{IC}(n,m)} \left[
\exists c \in [m] \text{ s.t. }
\forall b \in [m] \setminus \{c\},\;
|c \succ_{\mathcal{P}} b| \ge \alpha n-\left\lfloor\frac{k-1}{k}n\right\rfloor
\right]
= 0,
\]
which means that such a candidate $c$ is unlikely to exist.
 
\section{\texorpdfstring{$\alpha$}{Alpha}-Dominating Sets in the Limit}
\label{sec:dominating}
Now we turn to the stronger notion of $\alpha$-dominating sets and study the probability of their existence in a large impartial culture. Concretely, we will prove \cref{thm:alpha-dominating}. Recall that $g_{\alpha,k}(n,m)$ denotes the probability that there exists an $\alpha$-dominating committee $S \subseteq [m]$ of size $k$ in an impartial culture with $n$ voters and $m$ candidates.
\begin{theorem}
\label{thm:alpha-dominating}
For all constants $k \in \mathbb{N}^+$ and $\alpha \in [0, 1]$, there exists a number $N\in\mathbb{N}$ such that for every $n > N$, it holds that
\[
\lim_{m \to \infty} g_{\alpha,k}(n,m) =
\begin{cases}
1&\text{if }\alpha< \frac{k-1}{2k},\\[3pt]
0&\text{if }\alpha> \frac{k-1}{2k}.
\end{cases}
\]
\end{theorem}
To prove \cref{thm:alpha-dominating}, we consider the cases of $\alpha < \frac{k-1}{2k}$ and $\alpha > \frac{k-1}{2k}$ separately.

\paragraph{Case 1 (positive side): $\alpha \in \left[0, \frac{k-1}{2k}\right)$.}
In this case, we show that an $\alpha$-dominating set is likely to exist:
\[
\lim_{m\to\infty} g_{\alpha,k}(n,m)=1.
\]

Throughout this paper, we adopt the convention that $i \bmod k$ returns a number in $\{0, 1, \ldots, k - 1\}$. As an example, we have $(-5) \bmod 4 = 3$. We now define $r$-cyclic-threshold committees and will use them frequently in the case of $\alpha \in \left[0, \frac{k-1}{2k}\right)$. We illustrate $r$-cyclic-threshold committees in \cref{fig:cyclic}.

\begin{definition}[$r$-cyclic-threshold committee]
\label{def:cyclic}
Fix $k \geq 2$ and $r \in (0, m)$. Partition the $n$ voters into $k$ groups $V_1,\dots,V_k$, such that for each $i\in [k]$, the group
$V_i$ is $\{\lceil\frac{i-1}{k}n\rceil+1,\dots,\lceil\frac{i}{k}n\rceil\}$. For each $i\in[k]$, we use $\theta_i$ to denote the threshold $1-\exp\left(
        \frac{k-i}{k-1}\cdot 
        \frac{2(\ln r - \ln m)}{n}\cdot
        \frac{n/k}{\lceil n/k\rceil}
      \right)$.
We say that a committee\footnote{Here, in our definition, the tuple $(c_1,\dots,c_k)$ may contain repeated candidates. Conventionally, a committee is a subset (that is, without repetition) of the candidates. However, for convenience in our exposition, we define an $r$-cyclic-threshold committee as a tuple. In the end, we can simply remove duplications from it to obtain a committee without repetition.} $S = (c_1,\dots,c_k) \in [m]^k$ is an \emph{$r$-cyclic-threshold committee} if for every $i,j\in[k]$ and every voter $v\in V_i$, we have $s_{v,c_j} \geq \theta_{((i-j)\bmod k)+1}$.
\end{definition}

\definecolor{mycolor}{rgb}{0.65, 0.35, 0.45}

\begin{figure}[t]
\centering
\begin{tikzpicture}[scale=1.4]

\def\k{5}
\def\cellsize{0.9}

\foreach \i in {1,...,\k} {
  \pgfmathsetmacro{\clabel}{(\i+0.5)*\cellsize}
  \pgfmathsetmacro{\Vlabel}{-(\i+0.5)*\cellsize};
  \node at (\clabel, -0.4) [below] {$c_{\i}$};
  \node at (0.8,\Vlabel) [left] {$V_{\i}$};
}

\foreach \i in {1,...,\k} {
  \foreach \j in {1,...,\k} {
\pgfmathsetmacro{\d}{mod(\i-\j+\k,\k)}
\pgfmathsetmacro{\shade}{100 - 15*\d}
    \fill[mycolor!\shade] 
      (\j*\cellsize, -\i*\cellsize) 
      rectangle ++(\cellsize,-\cellsize);
    \draw 
      (\j*\cellsize, -\i*\cellsize) 
      rectangle ++(\cellsize,-\cellsize);
  }
}

\foreach \i in {1,...,\k} {
  \foreach \j in {1,...,\k} {

\pgfmathsetmacro{\idx}{int(mod(\i-\j+\k,\k)+1)}

\draw
      (\j*\cellsize, -\i*\cellsize)
      rectangle ++(\cellsize,-\cellsize);

\node at (\j*\cellsize+0.5*\cellsize, -\i*\cellsize-0.5*\cellsize)
      {$\theta_{\idx}$};
  }
}

\pgfmathsetmacro{\xlabel}{3*\cellsize}
\pgfmathsetmacro{\ylabel}{-(\k+2)*\cellsize}

\node at (\xlabel + 0.4, \ylabel + 0.4)
{\small darker color = higher threshold};

\end{tikzpicture}
\caption{Structure of an $r$-cyclic-threshold committee. Rows correspond to voter groups $V_i$, and columns correspond to committee members $c_j$. Each cell indicates the threshold $\theta_{((i-j)\bmod k)+1}$ required for voters in $V_i$ toward candidate $c_j$.}
\label{fig:cyclic}
\end{figure}
 
Let $\varepsilon=\frac{k-1}{2k}-\alpha$. Set a parameter $r_0=\ln m$. We will show that
\begin{itemize}
    \item an $r_0$-cyclic-threshold committee is likely to exist (stated in \cref{lem:find_S}), and
    \item it is likely that all $r_0$-cyclic-threshold committees are $\alpha$-dominating committees.
\end{itemize}
Hence, if we select any $r_0$-cyclic-threshold committee (which is likely to exist), it is probably also an $\alpha$-dominating committee.
\begin{lemma}[Existence of an $r$-cyclic-threshold committee]
\label{lem:find_S}
Let $r < m$ satisfy $r=\omega(1)$ as $m \to \infty$. In an impartial culture, with probability $1-o(1)$ as $m\to\infty$, there exists an $r$-cyclic-threshold committee.
\end{lemma}
\begin{proof}
The probability that there exists no $r$-cyclic-threshold committee is
\begin{align*}
&\ \Pr_{\mathcal{P}\sim\mathrm{IC}(n,m)}[\nexists S=(c_1,\dots,c_k)\text{ s.t.\@ $\forall i,j\in\{1,\dots,k\}$ and $\forall v\in V_i$, $s_{v,c_j}\ge\theta_{((i-j)\bmod k)+1}$}]\\
\le&\ \sum_{i=1}^k \Pr_{\mathcal{P}\sim\mathrm{IC}(n,m)}[\nexists c_i\text{ s.t.\@ $\forall j\in\{1,\dots,k\}$ and $\forall v\in V_i$, $s_{v,c_j}\ge\theta_{((i-j)\bmod k)+1}$}]\tag{union bound}\\
\le&\ k(1-(1-\theta_1)^{\lceil\frac{n}{k}\rceil}\cdots(1-\theta_k)^{\lceil\frac{n}{k}\rceil})^m.
\end{align*}
We denote $(1-\theta_1)^{\lceil\frac{n}{k}\rceil}\cdots(1-\theta_k)^{\lceil\frac{n}{k}\rceil}$ as $p$ and then calculate $p$ as
\begin{align*}
p=&\ \mathrm{exp}\left(\left\lceil\frac{n}{k}\right\rceil\sum_{i=1}^k\frac{k-i}{k-1}\cdot\frac{2(\ln r-\ln m)}{n}\cdot\frac{n/k}{\lceil n/k \rceil}\right)\\
=&\ \mathrm{exp}(\ln r-\ln m)=\frac{r}{m}.
\end{align*}

Since $r < m$ and $r=\omega(1)$ as $m \to \infty$, we have
\[
k(1-p)^m=k\left(1-\frac{r}{m}\right)^m\le k\exp\left(-\frac{r}{m} \cdot m\right) \leq \frac{k}{1 + r} = o(1)
\]
when $m\to\infty$. Hence, such a committee exists with probability $1 - o(1)$ when $m\to\infty$.
\end{proof}
Now that we have finished the proof of \cref{lem:find_S}, we move on to the next step and show that it is likely that all $r_0$-cyclic-threshold committees are $\alpha$-dominating committees.

Let $S=(c_1,\dots,c_k)$ be any $r_0$-cyclic-threshold committee. For any candidate $a\in [m]$, define ``times of win'' for $S$ against $a$ as
\[
t(S,a)\coloneqq\sum_{v\in [n]}\sum_{i=1}^k\mathbf{1}[c_i\succ_{R_v} a].
\]
If $t(S,a)\ge\left(\frac{k-1}{2k}-\varepsilon\right)nk$, then there will exist a candidate $c\in S$ such that $c$ wins at least $\left(\frac{k-1}{2k}-\varepsilon\right)n$ voters against $a$ by the pigeonhole principle. Recall that $\varepsilon=\frac{k-1}{2k}-\alpha$, and thus if $t(S,a)\ge\left(\frac{k-1}{2k}-\varepsilon\right)nk$ for all $a \in [m]$, then $S$ will be an $\alpha$-dominating committee. Therefore, we will prove that $t(S,a)\ge\left(\frac{k-1}{2k}-\varepsilon\right)nk$ simultaneously holds for all $a \in [m]$ with high probability.

Now we move on to show that it is unlikely that there exists a candidate $a \in [m]$ and an $r_0$-cyclic-threshold committee $S$ such that $t(S,a)<\left(\frac{k-1}{2k}-\varepsilon\right)nk$.
Assume there exists a candidate $a\in [m]$ and an $r_0$-cyclic-threshold committee $S$ such that $t(S,a)<\left(\frac{k-1}{2k}-\varepsilon\right)nk$ (and in the end, we will prove that this assumption is unlikely to hold). Let 
\[
\widetilde{\pi_1}(a) \coloneqq \frac{1}{n}|\{v\in [n]:s_{v,a}\ge\theta_1\}|,
\]
\[
\widetilde{\pi_2}(a) \coloneqq \frac{1}{n}|\{v\in [n]:\theta_1>s_{v,a}\ge\theta_2\}|,
\]
\[
\dots
\]
\[
\widetilde{\pi_k}(a) \coloneqq \frac{1}{n}|\{v\in [n]:\theta_{k-1}>s_{v,a}\ge\theta_k\}|,
\]
\[
\widetilde{\pi_{k+1}}(a) \coloneqq \frac{1}{n}|\{v\in [n]:s_{v,a}<\theta_k\}| = 0.
\]
Consider the ``times of win'' for $S$ against this $a$. Note that each voter counted in $\widetilde{\pi_i}(a)$ prefers at least $i - 1$ candidates in $S$ to the candidate $a$. Therefore, we have
\begin{equation}
\widetilde{\pi_2}(a)n+2\widetilde{\pi_3}(a)n+3\widetilde{\pi_4}(a)n+\dots+(k-1)\widetilde{\pi_k}(a)n\le t(S,a) < \left(\frac{k-1}{2k}-\varepsilon\right)nk,\label{eq:pi}
\end{equation}
where the last inequality is our previous assumption on $a$. Motivated by \cref{eq:pi}, we define the notion of a feasible sequence.
\begin{definition}[Feasible sequence]
\label{def:feasible}
We call $(\pi_1,\dots,\pi_k) \in \left([0, 1] \cap \frac{\mathbb{Z}}{n}\right)^k$ a \emph{feasible sequence} if
\[
\pi_1+\dots+\pi_k=1,
\qquad
\pi_2+2\pi_3+\dots+(k-1)\pi_k
    <\left(\frac{k-1}{2k}-\varepsilon\right)k.
\]
\end{definition}

Note that $(\widetilde{\pi_1}(a),\dots,\widetilde{\pi_k}(a))$ must be a feasible sequence.

Let $\widetilde{\bm{\pi}}(a)$ denote $(\widetilde{\pi_1}(a),\dots,\widetilde{\pi_k}(a))$ and consider any feasible sequence $\bm{\pi} = (\pi_1,\dots,\pi_k)$. We have
\begin{align*}
&\ \Pr_{\mathcal{P}\sim\mathrm{IC}(n,m)}[\exists a\text{ s.t.\@ }\widetilde{\bm{\pi}}(a)=\bm{\pi}]\\
\le&\ \sum_{a\in[m]} \Pr_{\mathcal{P}\sim\mathrm{IC}(n,m)}[\widetilde{\bm{\pi}}(a)=\bm{\pi}]\tag{union bound}\\
=&\ m \cdot \Pr_{\mathcal{P}\sim\mathrm{IC}(n,m)}[\widetilde{\bm{\pi}}(a_0)=\bm{\pi}].\tag{fix an arbitrary $a_0 \in [m]$; by symmetry}
\end{align*}
The event $\widetilde{\bm{\pi}}(a_0) = \bm{\pi}$ is equivalent to
\[
|\{v\in [n]:s_{v,a_0}\ge\theta_1\}| = \pi_1 n,
\]
\[
|\{v\in [n]:\theta_1>s_{v,a_0}\ge\theta_2\}| = \pi_2 n,
\]
\[
\dots
\]
\[
|\{v\in [n]:\theta_{k-1}>s_{v,a_0}\ge\theta_k\}| = \pi_k n.
\]
Since each voter contributes to exactly one of the $k$ intervals, we have
\begin{align*}
&\ \Pr_{\mathcal{P}\sim\mathrm{IC}(n,m)}[\widetilde{\bm{\pi}}(a_0)=\bm{\pi}]\\
=&\ \binom{n}{\pi_1 n, \pi_2 n, \ldots, \pi_k n} \cdot (1 - \theta_1)^{\pi_1 n} {(\theta_1 - \theta_2)}^{\pi_2 n} \cdots (\theta_{k-1} - \theta_k)^{\pi_k n}\\
\le&\ k^n \cdot (1 - \theta_1)^{\pi_1 n} {(1 - \theta_2)}^{\pi_2 n} \cdots (1 - \theta_k)^{\pi_k n}.
\end{align*}

If $\bm{\pi}$ is a feasible sequence, taking the logarithm of $(1-\theta_1)^{\pi_1n}\cdots(1-\theta_k)^{\pi_kn}$, we have
\begin{align*}
&\ \ln\left((1-\theta_1)^{\pi_1n}\cdots(1-\theta_k)^{\pi_kn}\right)\\
=&\ \sum_{i=1}^k\frac{k-i}{k-1}\cdot\frac{2(\ln r_0-\ln m)}{n}\cdot\frac{n/k}{\lceil n/k \rceil}\cdot \pi_in\tag{definition of $\theta_i$}\\
=&\ \frac{2(\ln r_0-\ln m)}{k-1}\cdot\frac{n/k}{\lceil n/k \rceil}\sum_{i=1}^k(k-i)\pi_i\\
=&\ \frac{2(\ln r_0-\ln m)}{k-1}\cdot\frac{n/k}{\lceil n/k \rceil}\left((k-1)\sum_{i=1}^k\pi_i-\sum_{i=1}^k(i-1)\pi_i\right).
\end{align*}
Then we apply \cref{def:feasible} and obtain
\begin{align*}
&\ \frac{2(\ln r_0-\ln m)}{k-1}\cdot\frac{n/k}{\lceil n/k \rceil}\left((k-1)\sum_{i=1}^k\pi_i-\sum_{i=1}^k(i-1)\pi_i\right)\\
<&\ \frac{2(\ln r_0-\ln m)}{k-1}\cdot\frac{n/k}{\lceil n/k \rceil}\left((k-1)-\left(\frac{k-1}{2k}-\varepsilon\right)k\right)\\
=&\ (\ln r_0-\ln m)\cdot\frac{n/k}{\lceil n/k \rceil}\left(1+\frac{2\varepsilon k}{k-1}\right).
\end{align*}
Recall that $r_0=\ln m$ and thus we have
\begin{align*}
&\ (1-\theta_1)^{\pi_1n}\cdots(1-\theta_k)^{\pi_kn}\\
=&\ \exp\left((\ln \ln m-\ln m)\cdot\frac{n/k}{\lceil n/k \rceil}\left(1+\frac{2\varepsilon k}{k-1}\right)\right)\\
=&\ \left(\frac{\ln m}{m}\right)^{\frac{n/k}{\lceil n/k \rceil}\left(1+\frac{2\varepsilon k}{k-1}\right)}.
\end{align*}
Since each voter contributes to exactly one of the $k$ intervals, we have
\[
|\{\bm{\pi}\mid \bm{\pi}\text{ is a feasible sequence} \land\bm{\pi}\text{ is in the image of }\widetilde{\bm{\pi}}\}|\le|\{\bm{\pi}\mid\bm{\pi}\text{ is in the image of }\widetilde{\bm{\pi}}\}|\le k^n.
\]
Therefore, we have (here and below, $\exists S$ is the shorthand for ``there exists an $r_0$-cyclic-threshold committee $S$'')
\begin{align*}
&\ \Pr_{\mathcal{P}\sim\mathrm{IC}(n,m)}\left[\exists a \exists S\text{ s.t.\@ }t(S,a)<\left(\frac{k-1}{2k}-\varepsilon\right)nk\right]\\
\le&\ \sum_{\substack{\bm{\pi}:\bm{\pi}\text{ is a feasible sequence} \\ \land\ \bm{\pi}\text{ is in the image of }\widetilde{\bm{\pi}}}}
\Pr_{\mathcal{P}\sim\mathrm{IC}(n,m)}[\exists a\text{ s.t.\@ }\widetilde{\bm{\pi}}(a)=\bm{\pi}]\\
\le&\ k^n\cdot m\cdot k^n \cdot (1-\theta_1)^{\pi_1n}\cdots(1-\theta_k)^{\pi_kn}\\
=&\ k^{2n} \cdot m \cdot\left(\frac{\ln m}{m}\right)^{\frac{n/k}{\lceil n/k \rceil}\left(1+\frac{2\varepsilon k}{k-1}\right)}\\
=&\ O_n\left(m^{1-\frac{n/k}{\lceil n/k \rceil}\left(1+\frac{2\varepsilon k}{k-1}\right)}\cdot (\ln m) ^{\frac{n/k}{\lceil n/k \rceil}\left(1+\frac{2\varepsilon k}{k-1}\right)}\right)
\end{align*}
in the limit of $m \to \infty$.
Notice that as long as $\frac{n}{n + k}\left(1+\frac{2\varepsilon k}{k-1}\right) > 1 + 2\varepsilon$ (that is, $n > (k^2 - k) \cdot \left(\frac{1}{2\varepsilon} + 1\right)$), we have 
\[
1-\frac{n/k}{\lceil n/k \rceil}\left(1+\frac{2\varepsilon k}{k-1}\right) < 1-\frac{n/k}{(n + k) / k}\left(1+\frac{2\varepsilon k}{k-1}\right) < -2\varepsilon,
\]
and thus we have
\[
\Pr_{\mathcal{P}\sim\mathrm{IC}(n,m)}\left[\exists a \exists S\text{ s.t.\@ }t(S,a)<\left(\frac{k-1}{2k}-\varepsilon\right)nk\right] = o_n(1)
\]
in the limit of $m \to \infty$ when $n > (k^2 - k) \cdot \left(\frac{1}{2\varepsilon} + 1\right)$.
Finally, we obtain 
\begin{align*}
&\ \lim_{m\to\infty}g_{\alpha,k}(n,m)\\
\ge&\ \lim_{m\to\infty}\Pr_{\mathcal{P}\sim\mathrm{IC}(n,m)}\left[\exists \text{an }r_0\text{-cyclic-threshold committee }S\land S\text{ is }\left(\frac{k-1}{2k}-\varepsilon\right)\!\text{-dominating}\right]\\
=&\ 1-\lim_{m\to\infty}\Pr_{\mathcal{P}\sim\mathrm{IC}(n,m)}[\nexists \text{an }r_0\text{-cyclic-threshold committee }S]\\
&\ \ \ -\lim_{m\to\infty}\Pr_{\mathcal{P}\sim\mathrm{IC}(n,m)}\left[\exists a \exists S\text{ s.t.\@ }t(S,a)<\left(\frac{k-1}{2k}-\varepsilon\right)nk\right]\\
=&\ 1.
\end{align*}

\paragraph{Case 2 (negative side): $\alpha \in \left(\frac{k-1}{2k}, 1\right]$.}
Here we show that for a sufficiently large $n$,
\[
\lim_{m\to\infty} g_{\alpha,k}(n,m)=0.
\]

Given the number of candidates $m$ and the number of voters $n$, we discretize the possible score vectors into the set $\mathrm{DS}$:
\[
\mathrm{DS}\coloneqq\left\{\bm{\beta}=(\beta_v)_{v\in [n]}: \forall v\in [n], \beta_v\in\left\{1-\left(\frac{2\ln m}{m}\right)^{\frac{n^2+1}{n^2}\cdot\frac{p}{n^2}}\right\}_{p\in[n^2]}\right\}.
\]
Additionally, for simplicity, we define the function $B(\bm{\beta})$ as a ``score summary'' of $\bm{\beta}$:
\[
B(\bm{\beta})\coloneqq\left(\prod_{v\in [n]}(1-\beta_v)\right)^{\frac{n-1}{n}} \in [0, 1].
\]
Intuitively, a ``better'' candidate should have a lower score summary.

We state the following lemmas and will prove them later in \cref{sec:skipped}. First, we show that we can probably discretize all scores into $\mathrm{DS}$ so that no candidate has a score summary that is ``too good.''
\begin{lemma}
\label{lem:discretize}
Given $m$ candidates, $n$ voters,
and a score matrix $\mathbf{S}=(s_{v,a})_{v\in [n],a\in [m]}\sim \mathrm{Unif}[0,1]^{n\times m}$,
we use $A_{\mathrm{D}}$ to denote the event that for each candidate $a\in [m]$, there exists $\bm{\beta}\in \mathrm{DS}$ such that
\[
B(\bm{\beta})>\left(\frac{2\ln m}{m}\right)^{\frac{n^4-1}{n^4}} \text{ and } \left|\{v\in [n]: s_{v,a}<\beta_v\}\right|=n.
\]
Then we have
\[
\lim_{m\to\infty}\Pr_{\mathbf{S}\sim\mathrm{Unif}[0,1]^{n\times m}}\left[A_{\mathrm{D}}\right]=1.
\]
\end{lemma}

Next, we show that for any $k$ discretized score vectors that are not ``too good,'' there exists another discretized score vector that is also not ``too good'' but can dominate each of the given ones in many dimensions.
\begin{lemma}
\label{lem:minimax}
Fix any $m$ and $n$, and fix any $k$ vectors $\bm{\beta_1},\bm{\beta_2},\ldots,\bm{\beta_k}\in \mathrm{DS}$ where $\bm{\beta_i}=(\beta_{i,v})_{v\in [n]}$ for each $i\in[k]$.
If, for each $i\in[k]$, we have
\[
B(\bm{\beta_i})>\left(\frac{2\ln m}{m}\right)^{\frac{n^4-1}{n^4}},
\]
then there exists $\bm{\beta_0}=(\beta_{0,v})_{v\in [n]}\in \mathrm{DS}$ such that
\[
B(\bm{\beta_0})>\frac{2\ln m}{m} \text{ and } \forall i\in[k], \left|\{v\in [n]: \beta_{0,v}\ge \beta_{i,v}\}\right|\ge\left(\frac{k+1}{2k}-n^{-\frac{1}{3}}\right)\cdot n,
\]
when
\[
\frac{\exp\left(2n^{\frac13}\right)}{n^4}>k.
\]
\end{lemma}

Finally, we show that for any discretized score vector that is not ``too good,'' there is probably a candidate that dominates the given score vector in almost all dimensions.
\begin{lemma}
\label{lem:good-rounding}
Given $m$ candidates, $n$ voters,
and a score matrix \(\mathbf{S}=(s_{v,a})_{v\in [n],a\in [m]}\sim \mathrm{Unif}[0,1]^{n\times m}\),
we use $A_{\mathrm{R}}$ to denote the event that for every $\bm{\beta}\in \mathrm{DS}$,
if $B(\bm{\beta})>\frac{2\ln m}{m}$,
then there exists a candidate $a\in [m]$ such that
\[
\left|\{v\in [n]: s_{v,a}>\beta_v\}\right|\ge n-1.
\]
Then we have
\[
\lim_{m\to\infty}\Pr_{\mathbf{S}\sim\mathrm{Unif}[0,1]^{n\times m}}\left[A_{\mathrm{R}}\right]=1.
\]
\end{lemma}

By \cref{lem:discretize}, 
when the event $A_{\mathrm{D}}$ occurs,
for any $k$-size committee $S=\{a_1,a_2,\ldots,a_k\}\subseteq C$, 
there exist $k$ deterministic score vectors $\bm{\beta_1},\bm{\beta_2},\ldots,\bm{\beta_k}\in \mathrm{DS}$,
where $\bm{\beta_i}=(\beta_{i,v})_{v\in [n]}$ for each $i\in[k]$,
such that for any $i\in[k]$,
\[
B(\bm{\beta_i})>\left(\frac{2\ln m}{m}\right)^{\frac{n^4-1}{n^4}} \text{ and } \left|\{v\in [n]: s_{v,a_i}<\beta_{i,v}\}\right|=n.
\]
By \cref{lem:minimax},
when
\[
\frac{\exp\left(2n^{\frac13}\right)}{n^4}>k,
\]
there exists a deterministic score vector $\bm{\beta_0}=(\beta_0,v)_{v\in [n]}\in \mathrm{DS}$ such that
\[
B(\bm{\beta_0})>\frac{2\ln m}{m} \text{ and } \forall i\in[k], \left|\{v\in [n]: \beta_{0,v}\ge \beta_{i,v}\}\right|\ge\left(\frac{k+1}{2k}-n^{-\frac{1}{3}}\right)\cdot n.
\]
By \cref{lem:good-rounding},
when the event $A_{\mathrm{R}}$ occurs,
there exists a candidate $a_0\in [m]$ such that
\[
\left|\{v\in [n]: s_{v,a_0}>\beta_{0,v}\}\right|\ge n-1.
\]

Hence, for any $i\in[k]$, we have
\[
\left|\{v\in [n]: s_{v,a_0}>s_{v,a_i}\}\right|\ge \left(\frac{k+1}{2k}-n^{-\frac{1}{3}}\right)\cdot n-1=\left(\frac{k+1}{2k}-n^{-\frac{1}{3}}-\frac{1}{n}\right)\cdot n,
\]
which implies that the committee $S$ is not $\alpha$-dominating for any
\[
\alpha\in\left(\frac{k-1}{2k}+n^{-\frac{1}{3}}+\frac{1}{n},1\right].
\]
By the union bound,
\[
g_{\alpha,k}(n,m)\le \Pr_{\mathbf{S}\sim\mathrm{Unif}[0,1]^{n\times m}}\left[\neg A_{\mathrm{D}}\right]+\Pr_{\mathbf{S}\sim\mathrm{Unif}[0,1]^{n\times m}}\left[\neg A_{\mathrm{R}}\right],
\]
for any
\[
\alpha\in\left(\frac{k-1}{2k}+n^{-\frac{1}{3}}+\frac{1}{n},1\right].
\]
By \cref{lem:discretize,lem:good-rounding}, for any $\alpha\in\left(\frac{k-1}{2k},1\right]$ and a sufficiently large $n$, we have
\begin{align*}
& \lim_{m\to\infty}g_{\alpha,k}(n,m)\\
\le&\ \left(\lim_{m\to\infty}\Pr_{\mathbf{S}\sim\mathrm{Unif}[0,1]^{n\times m}}\left[\neg A_{\mathrm{D}}\right]\right)
+\left(\lim_{m\to\infty}\Pr_{\mathbf{S}\sim\mathrm{Unif}[0,1]^{n\times m}}\left[\neg A_{\mathrm{R}}\right]\right)\\
=&\ 0.\qedhere
\end{align*}
 
\subsection{Skipped Proofs of \texorpdfstring{\cref{lem:discretize,lem:minimax,lem:good-rounding}}{Lemmas for Approximately Dominating Sets}}
\label{sec:skipped}
Now we complete the proofs for \cref{lem:discretize,lem:minimax,lem:good-rounding}.

\begin{proof}[Proof of \cref{lem:discretize}]

For each candidate $a\in [m]$, we similarly define
\[
B(a) \coloneqq \left(\prod_{v\in [n]}(1-s_{v,a})\right)^{\frac{n-1}{n}},
\]
which is a random variable depending on the random scores of the candidate $a$.

For any candidate $a\in [m]$, we have
\begin{align*}
&\ \Pr_{\mathbf{S}\sim\mathrm{Unif}[0,1]^{n\times m}}\left[B(a)\le \left(\frac{2\ln m}{m}\right)^{\frac{n^2+1}{n^2}\cdot\frac{n-1}{n}}\right]\\
=&\ \Pr_{\mathbf{S}\sim\mathrm{Unif}[0,1]^{n\times m}}\left[\prod_{v\in [n]}(1-s_{v,a})\le \left(\frac{2\ln m}{m}\right)^{\frac{n^2+1}{n^2}}\right]\\
=&\ \Pr_{\mathbf{S}\sim\mathrm{Unif}[0,1]^{n\times m}}\left[\sum_{v\in [n]}-\ln(1-s_{v,a})\ge -\frac{n^2+1}{n^2}\ln\left(\frac{2\ln m}{m}\right)\right].
\end{align*}
We calculate the CDF of $x_v=-\ln(1-s_{v,a})$ as follows:
\[
\Pr_{s_{v,a}\sim\mathrm{Unif}[0,1]}[x_v > z]=\Pr_{s_{v,a}\sim\mathrm{Unif}[0,1]}[-\ln(1-s_{v,a}) > z] = \Pr_{s_{v,a}\sim\mathrm{Unif}[0,1]}[1-s_{v,a} < e^{-z}] = e^{-z}.
\]
Therefore, the CDF of $x_v$ is $1 - e^{-z}$, which means that $x_v$ follows an exponential distribution. Therefore, $\sum_{v\in [n]}x_v$ follows a Gamma distribution: for all $z \geq 0$,
\[
\Pr\left[\sum_{v\in[n]}x_v < z\right]=1-e^{-z}\sum_{t=0}^{n-1}\frac{z^t}{t!}.
\]
Therefore, we have
\begin{align*}
&\ \Pr_{\mathbf{s}\sim\mathrm{Unif}[0,1]^n}\left[B(\mathbf{s})\le\left(\frac{2\ln m}{m}\right)^{\frac{n^2+1}{n^2}\cdot\frac{n-1}{n}}\right]\\
= &\ \Pr\left[\sum_{v\in [n]}x_v\ge -\frac{n^2+1}{n^2}\ln\left(\frac{2\ln m}{m}\right)\right]\\
= &\ \left(\frac{2\ln m}{m}\right)^{\frac{n^2+1}{n^2}}\sum_{t=0}^{n-1}\frac{\left(-\frac{n^2+1}{n^2}\ln\left(\frac{2\ln m}{m}\right)\right)^t}{t!}\\
= &\ \widetilde{O}_n\left(m^{-\frac{n^2+1}{n^2}}\right).
\end{align*}

We use $A_{\mathrm{B}}$ to denote the event that for every candidate $a\in [m]$, we have 
$B(a)>\left(\frac{2\ln m}{m}\right)^{\frac{n^2+1}{n^2}\cdot\frac{n-1}{n}}$.
By the union bound, we have
\begin{align*}
&\ \Pr_{\mathbf{S}\sim\mathrm{Unif}[0,1]^{n\times m}}\left[A_{\mathrm{B}}\right]\\
=&\ 1-\Pr_{\mathbf{S}\sim\mathrm{Unif}[0,1]^{n\times m}}\left[\exists a\in [m] \text{ s.t.\@ } B(a)\le \left(\frac{2\ln m}{m}\right)^{\frac{n^2+1}{n^2}\cdot\frac{n-1}{n}}\right]\\
\ge&\ 1-m\cdot \widetilde{O}_n\left({m^{-\frac{n^2+1}{n^2}}}\right)\\
=&\ 1-\widetilde{O}_n\left(m^{-\frac{1}{n^2}}\right).
\end{align*}

When the event $A_{\mathrm{B}}$ occurs, for any candidate $a\in [m]$ and for any voter $v\in [n]$, we have
\[
1-s_{v,a}\ge B(a)^{\frac{n}{n-1}}>\left(\frac{2\ln m}{m}\right)^{\frac{n^2+1}{n^2}}.
\]
Hence, there exists a unique integer $p_{v,a}\in[n^2]$ such that
\[
\left(\frac{2\ln m}{m}\right)^{\frac{n^2+1}{n^2}\cdot\frac{p_{v,a}}{n^2}}<1-s_{v,a}\le \left(\frac{2\ln m}{m}\right)^{\frac{n^2+1}{n^2}\cdot\frac{p_{v,a}-1}{n^2}}.
\]
Thus, for any candidate $a\in [m]$, we find a discretized score vector
\[
\bm{\beta_a}=\left(1-\left(\frac{2\ln m}{m}\right)^{\frac{n^2+1}{n^2}\cdot\frac{p_{v,a}}{n^2}}\right)_{v\in [n]},
\]
such that
\[
\left|\left\{v\in [n]: s_{v,a}<1-\left(\frac{2\ln m}{m}\right)^{\frac{n^2+1}{n^2}\cdot\frac{p_{v,a}}{n^2}}\right\}\right|=n,
\]
and that
\begin{align*}
&\ \left(\prod_{v\in [n]}\left(\frac{2\ln m}{m}\right)^{\frac{n^2+1}{n^2}\cdot\frac{p_{v,a}}{n^2}}\right)^{\frac{n-1}{n}}\\
\ge&\ \left(\prod_{v\in [n]}\left((1-s_{v,a})\cdot \left(\frac{2\ln m}{m}\right)^{\frac{n^2+1}{n^2}\cdot\frac{1}{n^2}}\right)\right)^{\frac{n-1}{n}}\\
=&\ B(a)\cdot \left(\frac{2\ln m}{m}\right)^{\frac{n^2+1}{n^2}\cdot\frac{n-1}{n^2}}.
\end{align*}

By the definition of $A_{\mathrm{B}}$, we have
\[
B(a)>\left(\frac{2\ln m}{m}\right)^{\frac{n^2+1}{n^2}\cdot\frac{n-1}{n}},
\]
and therefore
\[
B(a)\cdot \left(\frac{2\ln m}{m}\right)^{\frac{n^2+1}{n^2}\cdot\frac{n-1}{n^2}}>\left(\frac{2\ln m}{m}\right)^{\frac{n^4-1}{n^4}}.
\]
Hence, the event $A_{\mathrm{D}}$ occurs as long as the event $A_{\mathrm{B}}$ occurs. Therefore, we have
\[
\lim_{m\to\infty}\Pr_{\mathbf{S}\sim\mathrm{Unif}[0,1]^{n\times m}}\left[A_{\mathrm{D}}\right]\ge \lim_{m\to\infty}\Pr_{\mathbf{S}\sim\mathrm{Unif}[0,1]^{n\times m}}\left[A_{\mathrm{B}}\right]=1. \qedhere
\]
\end{proof}

\begin{proof}[Proof of \cref{lem:minimax}]
We define a $k\times k$ matrix $\mathbf{D}\coloneqq(D_{i,j})_{i,j\in[k]}$ such that for any $i,j\in[k]$,
\[
D_{i,j}\coloneqq\left|\left\{v\in [n]: \beta_{i,v}\le \beta_{j,v}\right\}\right|.
\]

Hence for any $i\in[k]$, we have
\[
D_{i,i}=n,
\]
and for any $i,j\in[k]$ with $i\ne j$, we have
\[
D_{i,j}+D_{j,i}\ge n.
\]

We construct a random score vector $\bm{\beta^*}=(\beta^*_v)_{v\in [n]}$ as follow.
First choose $k$ non-negative real numbers $x_1,x_2,\ldots,x_k$ that satisfy $\sum_{i=1}^k x_i=1$,
and then randomly set $\beta^*_v=\beta_{i,v}$ with probability $x_i$ independently for each voter $v\in [n]$.

Then for any $i\in[k]$, we have
\[
\E\left[\left|\{v\in [n]: \beta^*_v\ge \beta_{i,v}\}\right|\right]=\sum_{v\in [n]}\Pr\left[\beta^*_v\ge \beta_{i,v}\right].
\]

Consider an indicator variable
\[
\mathbf{1}[\beta_{j,v}\ge \beta_{i,v}]=
\begin{cases}
1,&\text{if } \beta_{j,v}\ge \beta_{i,v};\\
0,&\text{otherwise}.
\end{cases}
\]
By the definition of $\beta_v^*$, we have
\[
\sum_{v\in [n]}\Pr\left[\beta^*_v\ge \beta_{i,v}\right]=\sum_{v\in [n]}\sum_{j=1}^k x_j\cdot \mathbf{1}[\beta_{j,v}\ge \beta_{i,v}].
\]
By the definition of $D_{i,j}$, we have
\[
D_{i,j}=\sum_{v\in [n]}\mathbf{1}[\beta_{j,v}\ge \beta_{i,v}].
\]
Then by substitution, we obtain
\[
\E\left[\left|\left\{v\in [n]: \beta^*_v\ge \beta_{i,v}\right\}\right|\right]=\sum_{j=1}^k x_j\cdot D_{i,j}.
\]

We can choose $x_1,x_2,\ldots,x_k$ to maximize the minimum value of $\sum_{j=1}^k x_j\cdot D_{i,j}$ over all $i\in[k]$.
By the minimax theorem, we have
\[
\max_{x_1,x_2,\ldots,x_k}\min_{i\in[k]}\sum_{j=1}^k x_j\cdot D_{i,j}=\min_{y_1,y_2,\ldots,y_k}\max_{j\in[k]}\sum_{i=1}^k y_i\cdot D_{i,j},
\]
where $y_1,y_2,\ldots,y_k$ are non-negative real numbers that satisfy $\sum_{i=1}^k y_i=1$.

Because the maximum of a set of numbers is no smaller than any convex combination of them, we have
\[
\max_{j\in[k]}\sum_{i=1}^k y_i\cdot D_{i,j}
\ge\sum_{j=1}^k y_j\sum_{i=1}^k y_i\cdot D_{i,j}
=\sum_{i=1}^k\sum_{j=1}^k y_i y_j D_{i,j}
=\frac{1}{2}\sum_{i=1}^k\sum_{j=1}^k y_i y_j (D_{i,j}+D_{j,i}).
\]

Since $D_{i,j}+D_{j,i}\ge n$ for any $i,j\in[k]$ with $i\ne j$ and $D_{i,i}=n$ for any $i\in[k]$, we have
\[
\frac{1}{2}\sum_{i=1}^k\sum_{j=1}^k y_i y_j (D_{i,j}+D_{j,i})
\ge\frac{1}{2}\sum_{i=1}^k\sum_{j=1}^k y_i y_j n+\frac{1}{2}\sum_{i=1}^k y_i^2 n
=\frac{n}{2}\left(\sum_{i=1}^{k}y_i\right)\left(\sum_{j=1}^{k}y_j\right)+\frac{n}{2}\sum_{i=1}^k y_i^2.
\]
Note that $\sum_{i=1}^{k}y_i=1$. Hence
\[
\frac{n}{2}\left(\sum_{i=1}^{k}y_i\right)\left(\sum_{j=1}^{k}y_j\right)+\frac{n}{2}\sum_{i=1}^k y_i^2 = \frac{n}{2}+\frac{n}{2}\sum_{i=1}^k y_i^2 \ge \frac{n}{2}+\frac{n}{2k}=\left(\frac{k+1}{2k}\right)\cdot n.\
\]
In other words, there exists a choice of $x_1,x_2,\ldots,x_k$ such that for any $i\in[k]$,
\[
\E\left[\left|\left\{v\in [n]: \beta^*_v\ge \beta_{i,v}\right\}\right|\right]\ge \left(\frac{k+1}{2k}\right)\cdot n.
\]

By Hoeffding's inequality applied to $n$ independent Bernoulli random variables, we obtain that for any $i\in[k]$,
\[
\Pr\left[\left|\left\{v\in [n]: \beta^*_v\ge \beta_{i,v}\right\}\right|<\left(\frac{k+1}{2k}-n^{-\frac{1}{3}}\right)\cdot n\right]\le \exp\left(-2n^{\frac{1}{3}}\right).
\]
Applying the union bound over all $i\in[k]$, we get
\[
\Pr\left[\exists i\in[k] \text{ s.t.\@ } \left|v\in [n]: \beta^*_v\ge \beta_{i,v}\right|<\left(\frac{k+1}{2k}-n^{-\frac{1}{3}}\right)\cdot n\right]\le k\cdot\exp\left(-2n^{\frac{1}{3}}\right).
\]
By definition of $B(\bm{\beta^*})$,
\begin{align*}
&\ \E\left[-\ln B(\bm{\beta^*})\right]\\
=&\ \frac{n-1}{n}\cdot \sum_{v\in [n]}\E\left[-\ln(1-\beta^*_v)\right]\\
=&\ \frac{n-1}{n}\cdot \sum_{v\in [n]}\sum_{i=1}^k x_i\cdot \left(-\ln(1-\beta_{i,v})\right)\\
=&\ \frac{n-1}{n}\cdot \sum_{i=1}^k x_i\cdot \sum_{v\in [n]}\left(-\ln(1-\beta_{i,v})\right)\\
=&\ \sum_{i=1}^k x_i\cdot \left(-\ln B(\bm{\beta_i})\right).
\end{align*}
Since for any $i\in[k]$,
\[
B(\bm{\beta_i})>\left(\frac{2\ln m}{m}\right)^{\frac{n^4-1}{n^4}},
\]
we have
\[
\E\left[-\ln B(\bm{\beta^*})\right]<\frac{n^4-1}{n^4}\cdot \left(-\ln\left(\frac{2\ln m}{m}\right)\right).
\]
Applying Markov’s inequality, we obtain
\begin{align*}
&\ \Pr\left[B(\bm{\beta^*})\le \frac{2\ln m}{m}\right]\\
=&\ \Pr\left[-\ln B(\bm{\beta^*})\ge -\ln\left(\frac{2\ln m}{m}\right)\right]\\
\le&\ \frac{n^4-1}{n^4}.
\end{align*}
Hence, combining the two events via the union bound,
\begin{align*}
&\ \Pr\left[\exists i\in[k] \text{ s.t.\@ } \left|\{v\in [n]: \beta^*_v\ge \beta_{i,v}\}\right|<\left(\frac{k+1}{2k}-n^{-\frac{1}{3}}\right)\cdot n \text{ or } B(\bm{\beta^*})\le \frac{2\ln m}{m}\right]\\
\le&\ k\cdot \exp\left(-2n^{\frac{1}{3}}\right)+\frac{n^4-1}{n^4}\\
=&\ 1-\left(\frac{1}{n^4}-k\cdot \exp\left(-2n^{\frac{1}{3}}\right)\right).
\end{align*}
When
\[
\frac{\exp\left(2n^{\frac13}\right)}{n^4}>k,
\]
the above probability is strictly less than 1: there exists a realization of $\bm{\beta^*}$ such that for any $i\in[k]$, it holds that
\[
\left|\{v\in [n]: \beta^*_v\ge \beta_{i,v}\}\right|\ge \left(\frac{k+1}{2k}-n^{-\frac{1}{3}}\right)\cdot n,
\]
and
\[
B(\bm{\beta^*})>\frac{2\ln m}{m}.\qedhere
\]

\end{proof}

\begin{proof}[Proof of \cref{lem:good-rounding}]

Consider a discretized score vector $\bm{\beta}=(\beta_v)_{v\in [n]}$.
Then
\begin{align*}
&\ \Pr_{\mathbf{S}\sim\mathrm{Unif}[0,1]^{n\times m}}\left[\exists a\in [m] \text{ s.t.\@ } \left|\{v\in [n]: s_{v,a}>\beta_v\}\right|\ge n-1\right]\\
=&\ 1-\Pr_{\mathbf{S}\sim\mathrm{Unif}[0,1]^{n\times m}}\left[\forall a\in [m], \left|\{v\in [n]: s_{v,a}>\beta_v\}\right|<n-1\right].
\end{align*}
Since the candidates' scores are independent,
\begin{align*}
&\ \Pr_{\mathbf{S}\sim\mathrm{Unif}[0,1]^{n\times m}}\left[\forall a\in [m], \left|\{v\in [n]: s_{v,a}>\beta_v\}\right|<n-1\right]\\
=&\ \prod_{a\in [m]}\Pr_{\mathbf{S}\sim\mathrm{Unif}[0,1]^{n\times m}}\left[\left|\{v\in [n]: s_{v,a}>\beta_v\}\right|<n-1\right]\\
=&\ \prod_{a\in [m]}\left(1-\Pr_{\mathbf{S}\sim\mathrm{Unif}[0,1]^{n\times m}}\left[\left|\{v\in [n]: s_{v,a}>\beta_v\}\right|\ge n-1\right]\right).
\end{align*}
Let $v^*\in[n]$ be the voter that has the maximum $\beta_v$. Then,
\begin{align*}
&\ \Pr_{\mathbf{S}\sim\mathrm{Unif}[0,1]^{n\times m}}\left[\left|\{v\in [n]: s_{v,a}>\beta_v\}\right|\ge n-1\right]\\
\ge&\ \Pr_{\mathbf{S}\sim\mathrm{Unif}[0,1]^{n\times m}}\left[\forall v\in [n]\setminus\{v^*\}, s_{v,a}>\beta_v\right].
\end{align*}
By independence of the scores,
\begin{align*}
&\ \Pr_{\mathbf{S}\sim\mathrm{Unif}[0,1]^{n\times m}}\left[\forall v\in [n]\setminus\{v^*\}, s_{v,a}>\beta_v\right]\\
=&\ \prod_{v\in [n]\setminus\{v^*\}}\Pr_{\mathbf{S}\sim\mathrm{Unif}[0,1]^{n\times m}}\left[s_{v,a}>\beta_v\right]\\
=&\ \prod_{v\in [n]\setminus\{v^*\}}(1-\beta_v).
\end{align*}
Because $\beta_{v^*}\ge \beta_v$ for all $v\in [n]$, we have
\[
1-\beta_{v^*}\le \left(\prod_{v\in [n]}(1-\beta_v)\right)^\frac{1}{n},
\]
and therefore
\[
\prod_{v\in [n]\setminus\{v^*\}}(1-\beta_v)\ge \left(\prod_{v\in [n]}(1-\beta_v)\right)^{\frac{n-1}{n}}=B(\bm{\beta}).
\]
We then obtain
\begin{align*}
&\ \Pr_{\mathbf{S}\sim\mathrm{Unif}[0,1]^{n\times m}}\left[\exists a\in [m] \text{ s.t.\@ } \left|\{v\in [n]: s_{v,a}>\beta_v\}\right|\ge n-1\right]\\
\ge&\ 1-\left(1-B(\bm{\beta})\right)^m.
\end{align*}
If $B(\bm{\beta})>\frac{2\ln m}{m}$, then
\[
1-\left(1-B(\bm{\beta})\right)^m>1-\left(1-\frac{2\ln m}{m}\right)^m\ge 1-e^{-2\ln m}=1-\frac{1}{m^2}.
\]
Applying the union bound over all discretized score vectors in $\mathrm{DS}$, we have
\begin{align*}
&\ \Pr_{\mathbf{S}\sim\mathrm{Unif}[0,1]^{n\times m}}\left[A_{\mathrm{R}}\right]\\
\ge&\ 1-\sum_{\bm{\beta}\in \mathrm{DS} : B(\bm{\beta})>\frac{2\ln m}{m}}\Pr_{\mathbf{S}\sim\mathrm{Unif}[0,1]^{n\times m}}\left[\forall a\in [m], \left|\{v\in [n]: s_{v,a}>\beta_v\}\right|<n-1\right]\\
\ge&\ 1-|\mathrm{DS}|\cdot \frac{1}{m^2}.
\end{align*}
Since $|\mathrm{DS}|=n^{2n}$, we have
\[
\Pr_{\mathbf{S}\sim\mathrm{Unif}[0,1]^{n\times m}}\left[A_{\mathrm{R}}\right]\ge 1-\frac{n^{2n}}{m^2}.
\]
Thus,
\[
\lim_{m\to\infty}\Pr_{\mathbf{S}\sim\mathrm{Unif}[0,1]^{n\times m}}\left[A_{\mathrm{R}}\right]=1.\qedhere
\]
\end{proof}
 
\section{Discussion}

In this work, we studied the existence probability for $\alpha$-winning and $\alpha$-dominating sets in an impartial culture with many voters and even more candidates. Here, we present a few potential extensions to our work.
\begin{itemize}
    \item How likely do $\alpha$-winning and $\alpha$-dominating sets exist at exactly the threshold ($\frac{k - 1}{k}$-winning and $\frac{k - 1}{2k}$-dominating)? We expect that solutions to this question will fundamentally require new technical tools beyond our work.
    \item We showed that the thresholds are different in the iterated limit of $\lim_{n \to \infty} \lim_{m \to \infty}$ and the reversed iterated limit of $\lim_{m \to \infty} \lim_{n \to \infty}$. What is the behavior when neither variable is sufficiently larger than the other? In particular, is the limit considered in our work the worst case?
    \item Can our results be generalized to other preference generation models (beyond impartial cultures) or other utility models (beyond ranking)?
    \item Our probabilistic construction improves the previously best-known worst-case impossibility results for $\alpha$-dominating sets. Is this probabilistic construction also useful for providing impossibility results for other problems?
\end{itemize}
We leave these questions for future work.
 
\paragraph{Acknowledgment.} This work was partly carried out while Yifan Lin and Shenyu Qin were visiting the DIMACS Center at Rutgers University. Lirong Xia acknowledges NSF 2450124, 2517733, and 2518373 for support.

\newcommand{\etalchar}[1]{$^{#1}$}

\bibliographystyle{alpha}

\begin{thebibliography}{CDN{\etalchar{+}}25}

\bibitem[ABK{\etalchar{+}}06]{Alon06:Dominating}
Noga Alon, Graham~R. Brightwell, Hal~A. Kierstead, Alexandr~V. Kostochka, and Peter Winkler.
\newblock Dominating sets in k-majority tournaments.
\newblock {\em J. Comb. Theory {B}}, 96(3):374--387, 2006.

\bibitem[Arr63]{Arrow63:Social}
Kenneth Arrow.
\newblock {\em Social choice and individual values}.
\newblock New Haven: Cowles Foundation, 2nd edition, 1963.
\newblock 1st edition 1951.

\bibitem[BCT25]{Bourneuf2025:Neighborhood}
Romain Bourneuf, Pierre Charbit, and St{\'e}phan Thomass{\'e}.
\newblock {A Dense Neighborhood Lemma: Applications of Partial Concept Classes to Domination and Chromatic Number}.
\newblock In {\em Proceedings of FOCS}, 2025.

\bibitem[BGS16]{Brandt2016:Analyzing}
Felix Brandt, Christian Geist, and Martin Strobel.
\newblock {Analyzing the Practical Relevance of Voting Paradoxes via Ehrhart Theory, Computer Simulations, and Empirical Data}.
\newblock In {\em Proceedings of AAMAS}, pages 385--393, 2016.

\bibitem[BHS19]{Brandt2019:Exploring}
Felix Brandt, Johannes Hofbauer, and Martin Strobel.
\newblock {Exploring the No-Show Paradox for Condorcet Extensions Using Ehrhart Theory and Computer Simulations}.
\newblock In {\em Proceedings of AAMAS}, pages 520--528, 2019.

\bibitem[Bla48]{Black48:Rationale}
Duncan Black.
\newblock On the rationale of group decision-making.
\newblock {\em Journal of Political Economy}, 56(1):23--34, 1948.

\bibitem[Box79]{Box1979:Robustness}
George E.~P. Box.
\newblock Robustness in the strategy of scientific model building.
\newblock In R.~L. Launer and G.~N. Wilkinson, editors, {\em Robustness in Statistics}, pages 201--236. Academic Press, 1979.

\bibitem[CDN{\etalchar{+}}25]{Caizergues2025:Probability}
Emma Caizergues, Fran{\c c}ois Durand, Marc Noy, {\'E}lie de~Panafieu, and Vlady Ravelomanana.
\newblock {Probability of a Condorcet Winner for Large Electorates: An Analytic Combinatorics Approach}.
\newblock arXiv:2505.06028, 2025.

\bibitem[CFH{\etalchar{+}}24]{Conitzer2024:Social}
Vincent Conitzer, Rachel Freedman, Jobst Heitzig, Wesley~H. Holliday, Bob~M. Jacobs, Nathan Lambert, Milan Moss{\'e}, Eric Pacuit, Stuart Russell, Hailey Schoelkopf, Emanuel Tewolde, and William~S. Zwicker.
\newblock {Social Choice Should Guide AI Alignment in Dealing with Diverse Human Feedback}.
\newblock In {\em Proceedings of ICML}, 2024.

\bibitem[CLR{\etalchar{+}}25]{Charikar2025:Six-Candidates}
Moses Charikar, Alexandra Lassota, Prasanna Ramakrishnan, Adrian Vetta, and Kangning Wang.
\newblock {Six Candidates Suffice to Win a Voter Majority}.
\newblock In {\em Proceedings of STOC}, 2025.

\bibitem[Con85]{Condorcet1785:Essai}
Marquis~de Condorcet.
\newblock {\em Essai sur l'application de l'analyse \`a la probabilit\'e des d\'ecisions rendues \`a la pluralit\'e des voix}.
\newblock Paris: L'Imprimerie Royale, 1785.

\bibitem[CRW26]{Charikar2026:Approximately}
Moses Charikar, Prasanna Ramakrishnan, and Kangning Wang.
\newblock {Approximately Dominating Sets in Elections}.
\newblock In {\em Proceedings of SODA}, 2026.

\bibitem[DKNS01]{Dwork01:Rank}
Cynthia Dwork, Ravi Kumar, Moni Naor, and D.~Sivakumar.
\newblock Rank aggregation methods for the web.
\newblock In {\em Proceedings of the 10th World Wide Web Conference}, pages 613--622, 2001.

\bibitem[DM21]{Diss2021:Evaluating}
Mostapha Diss and Vincent Merlin, editors.
\newblock {\em {Evaluating Voting Systems with Probability Models}}.
\newblock Studies in Choice and Welfare. Springer International Publishing, 2021.

\bibitem[DP70]{DeMeyer1970:The-Probability}
Frank DeMeyer and Charles~R. Plott.
\newblock {The Probability of a Cyclical Majority}.
\newblock {\em Econometrica}, 38(2):345--354, 1970.

\bibitem[ELP16]{Elkind2016:Preference}
Edith Elkind, Martin Lackner, and Dominik Peters.
\newblock Preference restrictions in computational social choice: recent progress.
\newblock In {\em Proceedings of the Twenty-Fifth International Joint Conference on Artificial Intelligence}, pages 4062---4065, 2016.

\bibitem[ELS15]{Elkind2015:Condorcet}
Edith Elkind, J{\'e}r{\^o}me Lang, and Abdallah Saffidine.
\newblock {Condorcet winning sets}.
\newblock {\em Social Choice and Welfare}, 44:493--517, 2015.

\bibitem[GATC16]{Green-Armytage2016:Statistical}
James Green-Armytage, T.~Nicolaus Tideman, and Rafael Cosman.
\newblock Statistical evaluation of voting rules.
\newblock {\em Social Choice and Welfare}, 46(1):183--212, 2016.

\bibitem[GF76]{Gehrlein1976:The-probability}
William~V. Gehrlein and Peter~C. Fishburn.
\newblock {The probability of the paradox of voting: A computable solution}.
\newblock {\em Journal of Economic Theory}, 13(1):14--25, 1976.

\bibitem[GL11]{Gehrlein2011:Voting}
William~V. Gehrlein and Dominique Lepelley.
\newblock {\em {Voting Paradoxes and Group Coherence: The Condorcet Efficiency of Voting Rules}}.
\newblock Springer, 2011.

\bibitem[JMW20]{Jiang2020:Approximately}
Zhihao Jiang, Kamesh Munagala, and Kangning Wang.
\newblock {Approximately stable committee selection}.
\newblock In {\em Proceedings of STOC}, 2020.

\bibitem[JRTT95]{Jones1995:Condorcet}
Bradford Jones, Benjamin Radcliff, Charles Taber, and Richard Timpone.
\newblock {Condorcet Winners and the Paradox of Voting: Probability Calculations for Weak Preference Orders}.
\newblock {\em The American Political Science Review}, 89(1):137--144, 1995.

\bibitem[KR05]{Krishnamoorthy2005:Condorcet}
M.~S. Krishnamoorthy and M.~Raghavachari.
\newblock {Condorcet Winner Probabilities - A Statistical Perspective}.
\newblock arXiv:math/0511140, 2005.

\bibitem[LX16]{Lang16:Voting}
J\'{e}r\^{o}me Lang and Lirong Xia.
\newblock {Voting in Combinatorial Domains}.
\newblock In Felix Brandt, Vincent Conitzer, Ulle Endriss, J\'{e}r\^{o}me Lang, and Ariel Procaccia, editors, {\em {Handbook of Computational Social Choice}}, chapter~9. Cambridge University Press, 2016.

\bibitem[May71]{May1971:Some}
Robert~M. May.
\newblock Some mathematical remarks on the paradox of voting.
\newblock {\em Behavioral Science}, 16(2):143--151, 1971.

\bibitem[NSL26]{Nguyen2026:A-few-good}
Thanh Nguyen, Haoyu Song, and Young-San Lin.
\newblock {A few good choices}.
\newblock In {\em Proceedings of SODA}, 2026.

\bibitem[NW68]{Niemi1968:A-mathematical}
Richard~G. Niemi and Herbert~F. Weisberg.
\newblock {A mathematical solution for the probability of the paradox of voting}.
\newblock {\em Behavioral Science}, 13(4):317--323, 1968.

\bibitem[Pit25]{Pittel2025:On-Likelihood}
Boris Pittel.
\newblock {On Likelihood of a Condorcet Winner for Uniformly Random and Independent Voter Preferences}.
\newblock arXiv:2506.20613, 2025.

\bibitem[Saa95]{Saari1995:Basic}
Donald~G. Saari.
\newblock {\em {Basic Geometry of Voting}}.
\newblock Springer, 1995.

\bibitem[Sau22]{Sauermann2022:On-the-probability}
Lisa Sauermann.
\newblock {On the probability of a Condorcet winner among a large number of alternatives}.
\newblock arXiv:2203.13713, 2022.

\bibitem[Sen66]{Sen1966:A-Possibility}
Amartya~K. Sen.
\newblock {A Possibility Theorem on Majority Decisions}.
\newblock {\em Econometrica}, 32(2):491--499, 1966.

\bibitem[SN26]{SongN26}
Haoyu Song and Thanh Nguyen.
\newblock Approximate core of participatory budgeting via {L}inhahl equilibrium.
\newblock Manuscript, 2026.

\bibitem[SS04]{Sertel2004:Strong}
Murat~R. Sertel and M.~Remzi Sanver.
\newblock {Strong equilibrium outcomes of voting games are the generalized Condorcet winners}.
\newblock {\em Social Choice and Welfare}, 22:331--347, 2004.

\bibitem[TRG03]{Tsetlin2003:The-impartial}
Ilia Tsetlin, Michel Regenwetter, and Bernard Grofman.
\newblock The impartial culture maximizes the probability of majority cycles.
\newblock {\em Social Choice and Welfare}, 21(3):387--398, 2003.

\bibitem[Xia20]{Xia2020:The-Smoothed}
Lirong Xia.
\newblock {The Smoothed Possibility of Social Choice}.
\newblock In {\em Proceedings of NeurIPS}, 2020.

\bibitem[Xia21]{Xia2021:Semi-Random}
Lirong Xia.
\newblock {The Semi-Random Satisfaction of Voting Axioms}.
\newblock In {\em Proceedings of NeurIPS}, 2021.

\end{thebibliography}

\appendix
\section{Results for Impartial Culture with Many Voters}
For completeness, we include a simple proof of the phase transition thresholds in the alternative setting where the parameter $m$ is fixed and $n \to \infty$.
\begin{proposition}
\label{prop:large_n}
For all constants $k \in \mathbb{N}^+$ and $\alpha \in [0, 1]$, for every $m \in \mathbb{N}^+$, it holds that
\[
\lim_{n \to \infty}\Pr_{\mathrm{IC}(n,m)}\left[\text{An $\alpha$-winning committee exists} \right] = \begin{cases}
1&\text{if }\alpha < \frac{k}{k + 1},\\[3pt]
0&\text{if }\alpha > \frac{k}{k + 1},
\end{cases}
\]
and
\[
\lim_{n \to \infty}\Pr_{\mathrm{IC}(n,m)}\left[\text{An $\alpha$-dominating committee exists} \right] = \begin{cases}
1&\text{if }\alpha < \frac{1}{2},\\[3pt]
0&\text{if }\alpha > \frac{1}{2}.
\end{cases}
\]
\end{proposition}
\begin{proof}
First, we prove the $\alpha$-winning results. Consider any fixed committee $S$ of $k$ candidates and any candidate $c \notin S$. The probability that a voter prefers at least one candidate in $S$ to the outsider $c$ is $\frac{k}{k + 1}$. Therefore, when $\alpha < \frac{k}{k + 1}$, the probability that the committee $S$ is $\alpha$-winning approaches $1$ by applying a simple concentration bound and then the union bound on all outsiders. When $\alpha > \frac{k}{k + 1}$, the probability that there is an $\alpha$-winning committee approaches $0$ by applying a concentration bound and then the union bound on all committees. 

Next, we prove the $\alpha$-dominating results. Consider any fixed committee $S$ of $k$ candidates and any candidate $c \notin S$. The probability that a voter prefers a fixed candidate in $S$ to the outsider $c$ is $\frac{1}{2}$. Therefore, when $\alpha < \frac{1}{2}$, the probability that the committee $S$ is $\alpha$-dominating approaches $1$ by applying a concentration bound and then the union bound on all outsiders. when $\alpha > \frac{1}{2}$, the probability that there is an $\alpha$-dominating committee approaches $0$ by applying a concentration bound and then the union bound on all committees and all committee members.
\end{proof}
 
\end{document}